\documentclass[10pt]{article}
\usepackage{graphicx}
\usepackage{amsmath}
\usepackage[per-mode=symbol,detect-weight]{siunitx}
\usepackage{amssymb}
\usepackage{epstopdf}
\usepackage{a4wide}
\usepackage{amsthm}
\newcolumntype{P}[1]{>{\centering\arraybackslash}p{#1}}
\newtheorem{example}{Example}
\usepackage{subfigure}
\usepackage{graphicx,graphics}
\usepackage{hyperref}
\usepackage[vlined,ruled,linesnumbered]{algorithm2e}
\newtheorem{theorem}{Theorem}[section]
\newtheorem{definition}{Definition}[section]
\newtheorem{lemma}{Lemma}[section]

\begin{document}
	\begin{center}
		{{\bf \large {\rm {\bf  Fourth order compact scheme for option pricing under Merton and Kou jump-diffusion models}}}}
	\end{center}
	\begin{center}
		{\textmd {{\bf Kuldip Singh Patel,}}}\footnote{\it Department of Mathematics, Indian Institute of Technology, Delhi, India, (kuldip@maths.iitd.ac.in)}
		{\textmd {{\bf Mani Mehra}}}\footnote{\it Department of Mathematics, Indian Institute of Technology, Delhi, India, {(mmehra@maths.iitd.ac.in)} }
	\end{center}
\begin{abstract}
In this article, a three-time levels compact scheme is proposed to solve the partial integro-differential equation governing the option prices under jump-diffusion models. In the proposed compact scheme, the second derivative approximation of unknowns is approximated by the value of unknowns and their first derivative approximations which allow us to obtain a tri-diagonal system of linear equations for the fully discrete problem. Moreover, consistency and stability of the proposed compact scheme are proved. Due to the low regularity of typical initial conditions, the smoothing operator is employed to ensure the fourth-order convergence rate. Numerical illustrations for pricing European options under Merton and Kou jump-diffusion models are presented to validate the theoretical results.
\end{abstract}
\begin{center}
\end{center}
\section{Introduction}\label{sec:intro}
\par F. Black and M. Sholes derived a partial differential equation (PDE) governing the option prices in the stock market assuming that the dynamics of the underlying asset are driven by geometric Brownian motion with constant volatility \cite{BlaS73}. Later, numerous studies found that these assumptions are inconsistent with the market price movements. Various approaches have been considered to overcome the shortcomings of Black-Scholes model. In one of these approaches, Merton extended the Black-Scholes model to incorporate the jumps into the dynamics of the underlying asset in order to determine the volatility skews and it is known as Merton jump-diffusion model \cite{Mer76}. In another approach, the volatility is considered to be a stochastic process and these models are known as stochastic volatility models \cite{Hull87, Hes93}. Apart from these, the volatility is also assumed to be a deterministic function of time and stock price. This so-called deterministic volatility function approach was pioneered in \cite{Dup94}. Bates combined the jump-diffusion model with stochastic volatility approach to capture the typical features of market option prices \cite{Bates96}. Moreover, Anderson and Andreasen combined the deterministic volatility function approach with jump-diffusion model and proposed a second-order accurate numerical method for valuation of options \cite{Ander00}. In contrast to Merton jump-diffusion model where jump sizes follow Gaussian distribution, Kou proposed another jump-diffusion model assuming that jump sizes have double exponential distribution and it is called as Kou jump-diffusion model \cite{Kou02}.
\par The prices of European options under Merton and Kou jump-diffusion models can be evaluated by solving a partial integro-differential equation (PIDE). Various numerical methods have been proposed by several authors to solve the PIDE accurately and efficiently. The viscosity solution of the PIDE is discussed in \cite{Briani04} and an explicit finite difference method is proposed to solve the PIDE with certain stability condition. Cont and Voltchkova proposed implicit-explicit (IMEX) scheme for pricing European and barrier options and proved the stability and convergence of the proposed scheme \cite{Cont05}. d'Halluin et al. proposed a second-order accurate implicit method which uses fast Fourier transform (FFT) for evaluating the convolution integral \cite{Halluin05}. They also proved the stability and the convergence of the fixed-point iteration method. An excellent comparison of various approaches for option pricing under jump-diffusion models is given in \cite{Duffy05}. An IMEX Runge-Kutta method for the time integration is proposed in \cite{BrNaRu07} to solve the integral term of the PIDE explicitly. This method provides high-order accuracy under certain time step restriction. Sachs and Strauss used some transformation technique to eliminate the convection term from the PIDE and proposed a second-order accurate finite difference method for the solution of new PIDE \cite{SaSt08}. A three-time levels second-order accurate implicit method using finite difference approximations is proposed in \cite{Kwon11} for European put options under jump-diffusion models. A family of IMEX time discretization schemes is proposed in \cite{SaTo14} to solve the PIDE under jump-diffusion models. They discussed the stability of the schemes via Fourier analysis and showed that the schemes are conditionally stable. A second-order accurate IMEX time semi-discretization scheme for pricing European and American options under Bates model is discussed in \cite{SaToSy14}. They explicitly treated the jump term using the second-order Adams-Bashforth method and rest of the terms are discretized implicitly using the Crank-Nicolson method. Recently, Kadalbajoo et al. proposed a second-order accurate IMEX schemes along with cubic B-spline collocation method for European option pricing under jump-diffusion models \cite{KadT15}. They discussed the stability, convergence and computational complexity of all schemes.
\par It is observed that the inclusion of more grid points in computation stencil in order to increase the accuracy of finite difference approximations becomes computationally expensive. Therefore, finite difference approximations have been developed using compact stencils (commonly known as compact finite difference approximations) at the expense of some complication in their evaluation. Compact finite difference approximations provide high-order accuracy and better resolution characteristics as compared to finite difference approximations for equal number of grid points \cite{Lele92}. Compact finite difference approximations have also been used for option pricing problems \cite{DurF04, TanGB08, HWSun11}. In particular, Lee and Sun considered the original PIDE as an auxiliary equation to derive a compact scheme for option pricing under jump-diffusion models \cite{HWSun11}. They used IMEX schemes for temporal semi-discretization to avoid the inverse of a dense matrix and Richardson extrapolation is applied to obtain the fourth-order convergence rate. Recently, a fourth-order accurate compact scheme for option pricing under Bates model is derived in \cite{DuPi17} and stability of the compact scheme is observed numerically. Moreover, compact schemes have been extensively studied for compressible flows problems \cite{FasM00} and for computational aeroacoustic problems \cite{TKSen03}. A detailed study about various order compact finite difference approximations is given in \cite{MKPACM17}.
\par In this article, a three-time levels compact scheme is proposed to solve the PIDE under jump-diffusion models. The novelty of the proposed compact scheme is that it does not require the original equation as an auxiliary equation unlike the compact scheme proposed in \cite{HWSun11}. Moreover, consistency and stability of the proposed three-time levels compact scheme are proved. Since initial conditions for jump-diffusion models have low regularity, the smoothing operator given in \cite{thomee70} is employed to smoothen the initial conditions in order to achieve the fourth-order convergence rate. Simpson's rule for numerical integration is used and a Toeplitz-like structure is obtained in order to use FFT for efficient matrix-vector multiplication. Moreover, the CPU times for proposed compact scheme and finite difference scheme are calculated for a given accuracy and it is shown that proposed compact scheme outperforms the finite difference scheme.
\par The rest of the paper is organized as follows. In Section~\ref{sec:conti_prob}, the continuous model problem is discussed. Fourth-order compact
finite difference approximations for first and second derivatives along with a brief discussion on Fourier analysis are discussed in Section~\ref{sec:compact}. In Section~\ref{sec:disc_prob}, three-time levels scheme for temporal semi-discretization and fourth-order compact approximations for spatial discretization are discussed for continuous PIDE. The consistency and stability of proposed three-time levels compact scheme are proved in Section~\ref{sec:analysis}. In Section~\ref{sec:numerical}, numerical examples are presented to validate the theoretical results.
\section{The Continuous Problem}
\label{sec:conti_prob}
In this section, the mathematical model for pricing European options under jump-diffusion models is discussed. First we introduce the L$\acute{e}$vy process in the following definition.
\begin{definition} Let $(\Omega, \mathcal{F}, \mathcal{F}_{t}, \mathbb{P})$ be a probability space with filtration $\mathcal{F}_{t}$. A stochastic process $(X_{t})_{t\geq 0}$ is a L$\acute{e}$vy process on $(\Omega, \mathcal{F}, \mathcal{F}_{t}, \mathbb{P})$ if
	\begin{enumerate}
		\item $X_{0}=0$ almost surely.
		\item It is stochastically continuous i.e. for all $a>0$ and for all $s\geq0$
		\[
		\lim_{t \rightarrow s}P(|X(t)-X(s)|>a)=0.
		\]
		\item For any $s, t \geq0$, the distribution of $X_{t+s}-X_{s}$ does not depend on $s$ (stationary increments).
		\item For any $s\geq1$ and $0 \leq t_{0} < t_{1}<...<t_{s}$, the random variables $X_{t_{0}}$, $X_{t_{1}}-X_{t_{0}}$,..., $X_{t_{s}}-X_{t_{s-1}}$ are independent (independent increments).
		\item The sample paths of $X_{t}$ are right continuous with left limits almost surely.
	\end{enumerate}
\end{definition}
The Poisson process $((N_{t})_{t\geq 0})$ and the Brownian motion $((W_{t})_{t\geq 0})$ are the examples of L$\acute{e}$vy processes. Suppose that the stock price process follows an exponential jump-diffusion model $S_{t}=S_{0}e^{rt+X_{t}}$, where $r$ is the risk-free interest rate, $S_{0}$ is the stock price at $t=0$ and $\left(X_{t}\right)_{t\geq 0}$ is a jump-diffusion L$\acute{e}$vy process. The jump-diffusion L$\acute{e}$vy process $\left(X_{t}\right)_{t\geq 0}$ is defined as
\begin{equation}
\label{eq:jdprocess}
X_{t}:=at+\sigma W_{t}+\sum_{i=1}^{N_{t}} G_{i},
\end{equation}
where $a$ and $\sigma > 0$ are real constants and $G_{i}$ are  independent and identically distributed random variable with density function $g(x)$. Furthermore, $W_{t}$, $G_{i}$, and $N_{t}$ are assumed to be
mutually independent. In Merton jump-diffusion model, $G_{i}$ follows Gaussian distribution with density function
\begin{equation}
\label{eq:mertondensity}
g(x)=\frac{1}{\sqrt{2\pi\sigma_{J}^2}}e^{-\frac{(x-\mu_{J})^2}{2\sigma_{J}^2}},
\end{equation}
where $\mu_{J}$ and $\sigma_{J} > 0$ represents mean and standard deviation respectively. In case of Kou jump-diffusion model, $G_{i}$ follows double exponential distribution with density function
\begin{equation}
\label{eq:koudensity}
g(x)=(1-p)\lambda_{+}e^{-\lambda+x}1_{x\geq0}+p\lambda_{-}e^{\lambda-x}1_{x<0},
\end{equation}
where $1_{A}$ is the indicator function with respect to a set $A$, $\lambda_{-}>0$, $\lambda_{+}>1$ and $0\leq p \leq 1$. The price of
European options under jump-diffusion models ($V(S,t)$) is obtained by solving a PIDE which is discussed in the following theorem \cite{Kwon11}.
\begin{theorem} Let the L$\acute{e}$vy process $\left(X_{t}\right)_{t\geq 0}$ has the L$\acute{e}$vy triplet $(\sigma^2, \gamma, \nu)$, where $\sigma >0$, $\gamma \in \mathbb{R}$
	and $\nu$ is the L$\acute{e}$vy measure. If
	\[
	\sigma>0 \:\: or \;\: \exists \:\: \beta \in (0,2)\:\:\:\: \mbox{such that} \:\:\: \liminf \limits_{\epsilon\rightarrow 0} \epsilon^{-\beta}\int_{-\epsilon}^{\epsilon} |x|^2 \nu(dx) >0,
	\]
	then the value of European option with the payoff function $Z(S_{T})$ is obtained by $V(S,t)$, where
	\[
	V:[0,\infty)\times[0,T] \rightarrow \mathbb{R},
	\]
	\[
	(S,t)\mapsto V(S,t)=e^{-r(T-t)}\mathbb{E}[Z(S_{T})|S_{t}=S],
	\]
	is a continuous map on $[0,\infty)\times[0,T]$, $C^{1,2}$ on $(0,\infty)\times(0,T)$, and satisfies the following PIDE
	\begin{equation}
	\begin{split}
	\label{eq:PIDE}
	-\frac{\partial V}{\partial t}(S,t)&=\frac{\sigma^2 S^2}{2}\frac{\partial^2 V}{\partial S^2}(S,t)+rS\frac{\partial V}{\partial S}(S,t)-rV(S,t)\\
	&+\int_{\mathbb{R}}^{}\left[ V(Se^{x},t)-V(S,t)-S(e^x-1)\frac{\partial V}{\partial S}(S,t)\right]\nu(dx),
	\end{split}
	\end{equation}
	on $(0,\infty)\times[0,T)$ with the final condition
	\[
	V(S,T)=Z(S) \:\:\:\: \forall \:\: S>0.
	\]
\end{theorem}
Let us consider the following transformation in the above PIDE (\ref{eq:PIDE})
\[
\tau=T-t, \:x=ln\left(\frac{S}{S_{0}}\right) \:\mbox{and} \:u(x,\tau) = V(S_{0}e^{x},T-\tau).
\]
Then, $u(x,\tau)$ is the solution of the following PIDE with constant coefficients
\begin{equation}
\begin{split}
\label{eq:pidefinal}
\frac{\partial u}{\partial \tau}(x,\tau)&=\mathbb{L}u,\: (x,\tau)\in \:(-\infty,\infty)\times(0,T],\\
u(x,0) &= f(x)\:\:\: \forall \:\:\: x \in (-\infty, \infty),
\end{split}
\end{equation}
where
\small
\begin{equation}
\label{eq:operator}
\mathbb{L}u=\frac{\sigma^2}{2}\frac{\partial^2 u}{\partial x^2}(x,\tau)+\left(r-\frac{\sigma^2}{2}-\lambda \zeta\right)\frac{\partial u}{\partial x}(x,\tau)-(r+\lambda)u(x,\tau)+\lambda \int_{\mathbb{R}}^{} u(y,\tau)g(y-x)dy,
\end{equation}
\normalsize
$\lambda$ is the intensity of the jump sizes and $\zeta$ = $\int_{\mathbb{R}}^{} (e^x-1)g(x)dx$. The initial condition for European call options is
\begin{equation}
\label{eq:initial_call}
f(x)= max(S_{0}e^{x}-K,0) \:\:\: \forall \:\:\:x \in \mathbb{R},
\end{equation}
and the equations describing the asymptotic behaviour of European call options are
\begin{equation}
\label{eq:boundary_call}
\lim_{x\rightarrow -\infty}u(x,\tau)=0 \:\:\:\:\:\mbox{and} \:\:\:\:\:\lim_{x\rightarrow \infty}[u(x,\tau)-(S_{0}e^{x}-Ke^{-r\tau})]=0.
\end{equation}
Similarly, the initial condition for European put options is
\begin{equation}
\label{eq:initial_put}
f(x)= max(K-S_{0}e^{x},0) \:\:\: \forall \:\:\:x \in \mathbb{R},
\end{equation}
and the asymptotic behaviour of European put options is described as
\begin{equation}
\label{eq:boundary_put}
\lim_{x\rightarrow -\infty}[u(x,\tau)-(Ke^{-r\tau}-S_{0}e^{x})]=0\:\:\:\:\:\mbox{and} \:\:\:\:\:\lim_{x\rightarrow \infty}u(x,\tau)=0.
\end{equation}
\section{Fourth-Order Compact Finite Difference Approximations for First and Second Derivatives}
\label{sec:compact}
Compact finite difference approximations for first and second derivatives are discussed in this section. From Taylor series expansion, second-order accurate finite difference approximations for first and second derivatives can be written as
\begin{equation}
\label{eq:firstf_2}
\Delta_{x}u_{i}=\frac{u_{i+1}-u_{i-1}}{2\delta x},\:\:\:\:\Delta^2_{x}u_{i}=\frac{u_{i+1}-2u_{i}+u_{i-1}}{\delta x^2},
\end{equation}
where $u_{i}$ is the value of $u$ at a typical grid point $x_{i}$. Moreover, fourth-order accurate compact finite difference approximations for first and second derivatives \cite{Lele92} are
\begin{equation}
\label{eq:firstc_4}
\frac{1}{4}u_{x_{i-1}}+u_{x_{i}}+\frac{1}{4}u_{x_{i+1}}= \frac{1}{\delta x}\left[-\frac{3}{4}u_{i-1}+\frac{3}{4}u_{i+1}\right],
\end{equation}
\begin{equation}
\label{eq:pade2_4}
\frac{1}{10}u_{xx_{i-1}}+u_{xx_{i}}+\frac{1}{10}u_{xx_{i+1}} = \frac{1}{\delta x^2}\left[\frac{6}{5}u_{i-1}-\frac{12}{5}u_{i}+\frac{6}{5}u_{i+1}\right],
\end{equation}
where $u_{x_{i}}$, $u_{xx_{i}}$ are first and second derivatives of unknown $u$ at grid point $x_{i}$. If first derivative is also considered as a variable then from Equation~(\ref{eq:firstc_4}) we can write
\begin{equation}
\label{eq:pade3_4}
\frac{1}{4}u_{xx_{i-1}}+u_{xx_{i}}+\frac{1}{4}u_{xx_{i+1}}= \frac{1}{\delta x}\left[-\frac{3}{4}u_{x_{i-1}}+\frac{3}{4}u_{x_{i+1}}\right].
\end{equation}
Eliminating $u_{xx_{i-1}}$ and $u_{xx_{i+1}}$ from Equations~(\ref{eq:pade2_4}) and~(\ref{eq:pade3_4}), compact finite difference approximation for second derivative is
\begin{equation}
\label{eq:secondc1_4}
u_{xx_{i}}=2\frac{u_{i+1}-2u_{i}+u_{i-1}}{\delta x^2}-\frac{u_{x_{i+1}}-u_{x_{i-1}}}{2\delta x}.
\end{equation}
Substituting the values from Equation~(\ref{eq:firstf_2}) into Equation~(\ref{eq:secondc1_4}), we get
\begin{equation}
\label{eq:secondc_4}
u_{xx_{i}}=2\Delta^2_{x}u_{i}-\Delta_{x}u_{x_{i}}.
\end{equation}
It is observed that Equations~(\ref{eq:firstc_4}) and~(\ref{eq:secondc_4}) provide fourth-order accurate compact finite difference approximations for first and second derivatives. The value of  $u_{x_{i}}$ in Equation~(\ref{eq:secondc_4}) is obtained from Equation~(\ref{eq:firstc_4}). One-sided compact finite difference approximations are discussed in \cite{KPMI17} for non-periodic boundary conditions. It is shown in Figure~\ref{fig:grid} that lesser number of grid points are required with compact finite difference approximation as compared to finite difference approximation to obtain high-order accuracy.
\begin{figure}[h!]
	\begin{center}
		\includegraphics[trim = 0cm 20cm 0cm 3cm, clip, width=1.1\textwidth]{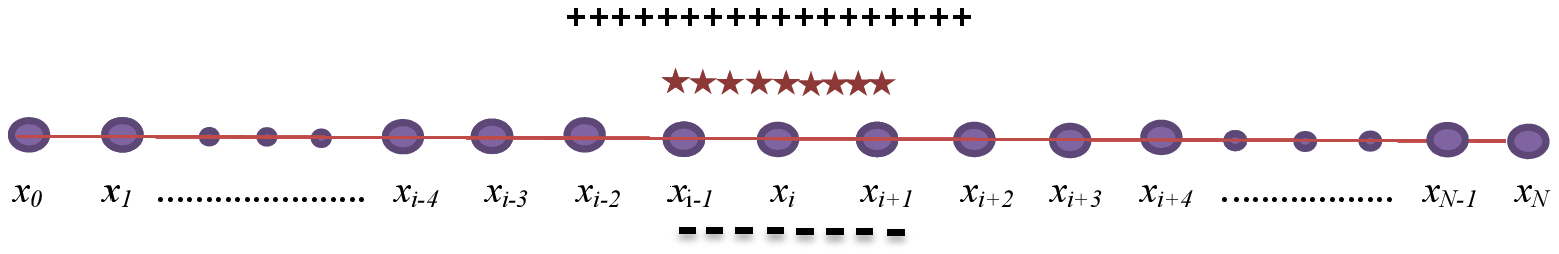}
		\vspace{-1.5 cm}
		\caption{Number of grid points required for first derivative approximation : (a). $\textbf{\textbf{-}} \:\textbf{\textbf{-}} \: \textbf{\textbf{-}}  \: \textbf{\textbf{-}} \:\textbf{\textbf{-}}  \:\textbf{\textbf{-}} \: \textbf{\textbf{-}} \: \textbf{\textbf{-}}$ $(x_{i-1}, x_{i}, x_{i+1})$ : $(\delta x^4)$ compact finite difference approximation, (b). $\star\star\star\star\star\star\star\star$ $(x_{i-1}, x_{i}, x_{i+1})$ : $O(\delta x^2)$ finite difference approximation and (c). $\textbf{\textbf{+}}\textbf{\textbf{+}}\textbf{\textbf{+}}\textbf{\textbf{+}}\textbf{\textbf{+}}\textbf{\textbf{+}}\textbf{\textbf{+}}\textbf{\textbf{+}}$ $(x_{i-2}, x_{i-1}, x_{i}, x_{i+1}, x_{i+2})$ : $O(\delta x^4)$ finite difference approximation.}
		\label{fig:grid}
	\end{center}
\end{figure}
\subsection{Fourier analysis}
\label{ssec:fourier}
In this section, the wave numbers and the modified wave numbers for first and second derivative approximations are discussed in brief. A detailed discussion on the resolution characteristics of various order compact finite difference approximations is given in \cite{Lele92}. The trial function for this one on a periodic domain is $u(x)=e^{I\omega x}$, where $I=\sqrt{-1}$ and $\omega$ is known as wavenumber. The relations between $\omega$ (wave number), $\omega'$ (modified wave number for first derivative) and $\omega''$ (modified wave number for second derivative) for finite difference approximations and proposed compact finite difference approximations are given in \cite{KPMA17}. The wave numbers versus modified wave numbers are plotted in Figures~\ref{fig:fou_first} and~\ref{fig:fou_second} for first and second derivative approximations and it is observed from figures that compact finite difference approximations have better resolution characteristics as compared to the finite difference approximations.
\begin{figure}
	\begin{center}
		\subfigure[]{%
			\includegraphics[scale=0.420]{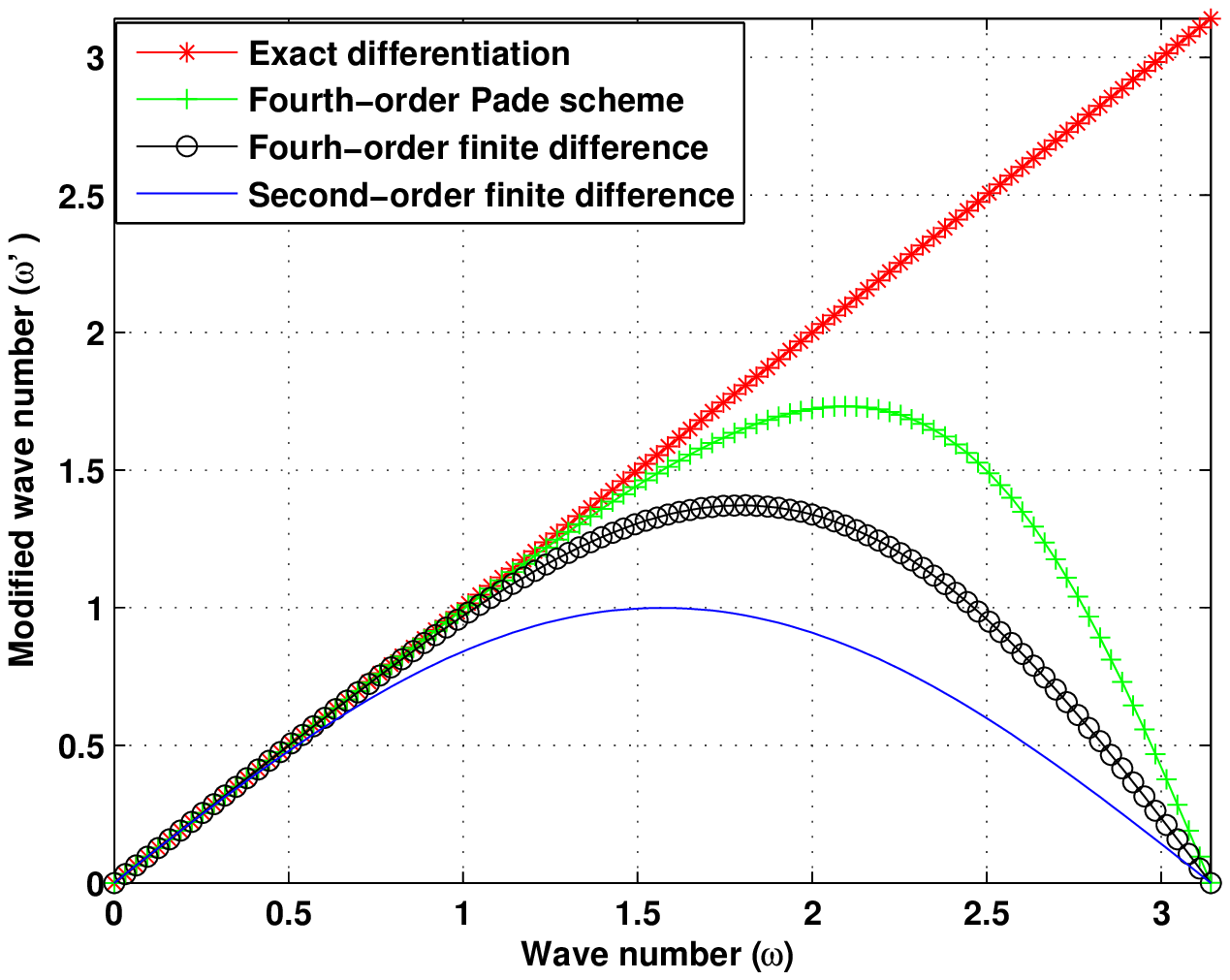}
			\label{fig:fou_first}}%
		\subfigure[]{%
			\includegraphics[scale=0.420]{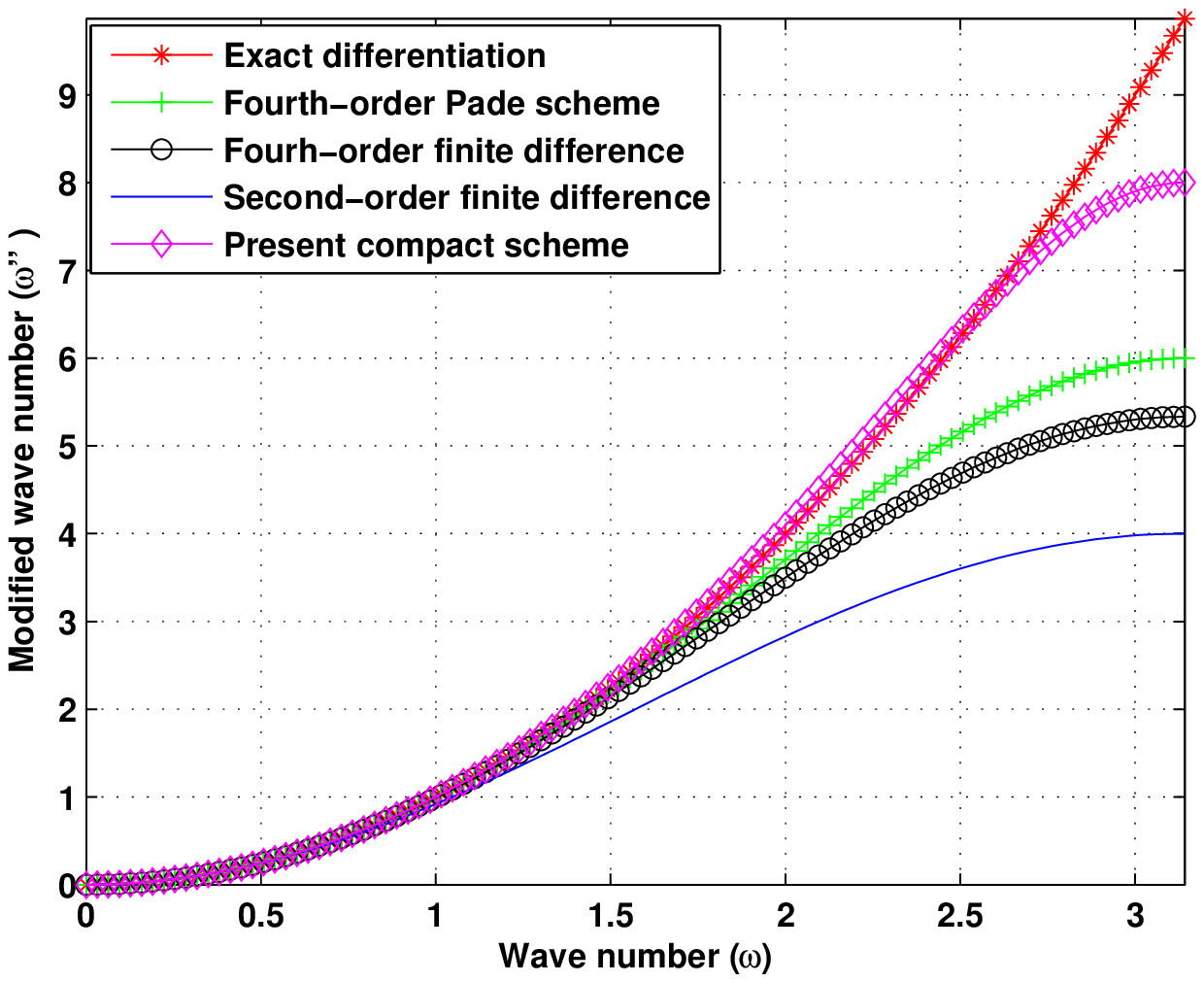}
			\label{fig:fou_second}}
		\caption{Modified wave number and wave number for various finite difference schemes: (a) First derivative approximation, (b). Second derivative
			approximation.}
		\label{fig:fou}
	\end{center}
\end{figure}
\section{The Fully Discrete Problem}
\label{sec:disc_prob}
\par The domain of the spatial variable is restricted to a bounded interval $\Omega=(-L,L)$ for some fixed real number $L$ in order to solve the PIDE (\ref{eq:pidefinal}) numerically. For given positive integers $M$ and $N$, let $\delta x=2L/N$
and $\delta \tau=T/M$ and in this way we define $x_{n}=-L+n\delta x$ $(n=0,1,....,N)$
and $\tau_{m}=m\delta \tau$ $(m=0,1,...,M)$. Let us first introduce the numerical approximation for $\frac{\partial u}{\partial \tau}$. The second-order accurate finite difference approximation for $\frac{\partial u}{\partial \tau}$ at each grid point $(x_{n},\tau_{m})$ is given by
\begin{equation}
\frac{\partial u}{\partial \tau}(x_{n},\tau_{m})=\frac{u^{m+1}_{n}-u^{m-1}_{n}}{2\delta \tau}+O(\delta \tau^2),\:\:\: \mbox{for}\:\:\: m \geq 1,
\end{equation}
where $u^{m}_{n}=u(x_{n},\tau_{m})$. The operator $\mathbb{L}$ in Equation~(\ref{eq:operator}) can be written as
\begin{equation}
\label{eq:l}
\mathbb{L}u(x,\tau)=\mathbb{D}u(x,\tau)+\mathbb{I}u(x,\tau)-(r+\lambda)u(x,\tau),
\end{equation}
where
\begin{equation}
\label{eq:all_operator}
\begin{split}
\mathbb{D}u(x,\tau) &= \frac{\sigma^2}{2}\frac{\partial^2 u}{\partial x^2}(x,\tau)+(r-\frac{\sigma^2}{2}-\lambda \zeta)\frac{\partial u}{\partial x}(x,\tau),\\
\mathbb{I}u(x,\tau) &= \lambda \int_{\mathbb{R}}^{}u(y,\tau)g(y-x)dy.
\end{split}
\end{equation}
Now, the numerical approximations for the differential operator $\mathbb{D}$ is discussed. If $\mathbb{D}_{\delta}$ represents the discrete approximations for the operator $\mathbb{D}$ then
\begin{equation*}
\mathbb{D}u^{m}_{n} \approx \mathbb{D_{\delta}}\left(\frac{u^{m+1}_{n}+u^{m-1}_{n}}{2}\right),
\end{equation*}
where
\begin{equation}
\label{eq:ddel}
\mathbb{D_{\delta}}{u^{m}_{n}}=\frac{\sigma^2}{2}u^{m}_{xx_{n}}+\left(r-\frac{\sigma^2}{2}-\lambda \zeta\right
)u^{m}_{x_{n}},
\end{equation}
$u^{m}_{x_{n}}$ and $u^{m}_{xx_{n}}$ are the first and second derivative approximations of $u_n^m$ respectively. Substituting the value of $u^{m}_{xx_{n}}$ from Equation~(\ref{eq:secondc_4}) into the Equation~(\ref{eq:ddel}), we get
\begin{equation}
\label{eq:ddelta}
\mathbb{D_{\delta}}{u^{m}_{n}}=\frac{\sigma^2}{2}\left(2\Delta^2_{x}u^{m}_{n}-\Delta_{x}u_{x_{n}}^{m}\right)+\left(r-\frac{\sigma^2}{2}-\lambda \zeta\right)u^{m}_{x_{n}}.
\end{equation}
In this way, we eliminate the compact finite difference approximations of second derivative using the unknowns and their first derivative approximations.
\par Now, the discrete approximation for the integral operator $\mathbb{I}$ using fourth-order accurate composite Simpson's rule is discussed. Integral operator
$\mathbb{I}u(x,\tau)$ given in Equation~(\ref{eq:all_operator}) is divided into two parts namely on $\Omega=(-L,L)$ and $\mathbb{R}\backslash\Omega$. The value of integration over $\mathbb{R}\backslash\Omega$ for Merton jump-diffusion model is given as
\begin{equation}
\label{eq:psimerton}
\Upsilon(x,\tau,L)=Ke^{-r\tau}\Phi\left(-\frac{x+\mu_{J}+L}{\sigma_{J}}\right)-S_{0}e^{x+\frac{\sigma^2_{J}}{2}+\mu_{J}}\Phi\left(-\frac{x+\sigma^2_{J}+\mu_{J}+L}{\sigma_{J}}\right),
\end{equation}
where $\Phi(x)$ is the cumulative distribution function of standard normal distribution. Similarly, the value of integration over $\mathbb{R}\backslash\Omega$ for Kou jump-diffusion model is given as
\begin{equation}
\label{eq:psikou}
\Upsilon(x,\tau,L)=K(1-p)e^{-r\tau-\lambda_{-}(L+x)}-S_{0}(1-p)\frac{\lambda_{-}}{\lambda_{-}+1}e^{-\lambda_{-}x-(\lambda_{-}+1)L}.
\end{equation}
The value of integral $\mathbb{I}u(x,\tau)$ on the interval $\Omega$ using composite Simpson's rule is given as
\small
\begin{equation}
\begin{split}
\label{eq:simp}
\int_{\Omega}^{}u(y,\tau_{m})g(y-x_{n})dy &= \frac{\delta x}{3}\left(u^{m}_{0}g_{n,0}+
4\sum_{i=1}^{\frac{N}{2}} u^{m}_{2i-1}g_{n,2i-1}+2\sum_{i=1}^{\frac{N}{2}-1} u^{m}_{2i}g_{n,2i}+u^{m}_{N}g_{n,N}\right)\\
&+O(\delta x^4),
\end{split}
\end{equation}
\normalsize
where $g_{n,i}=g(x_{i}-x_{n})$. In order to write the above integral approximation~(\ref{eq:simp}) in matrix-vector multiplication form, we define
\begin{center}
	$B_{g}=\frac{\delta x}{3}\left[
	\begin{array}{ccccc}
	4g(x_1-x_1) & 2g(x_2-x_1) & 4g(x_3-x_1) &  \dots  & 4g(x_{N-1}-x_1) \\
	4g(x_1-x_2) & 2g(x_2-x_2) & 4g(x_3-x_2) &  \dots  & 4g(x_{N-1}-x_2) \\
	4g(x_1-x_3) & 2g(x_2-x_3) & 4g(x_3-x_3) &  \dots  & 4g(x_{N-1}-x_3) \\
	\dots &   \dots & \dots & \dots & \dots \\
	4g(x_1-x_{N-1}) & 2g(x_2-x_{N-1}) & 4g(x_3-x_{N-1}) &  \dots  & 4g(x_{N-1}-x_{N-1}) \\
	\end{array}
	\right],$
\end{center}
\begin{equation*}
u^{m}=\begin{bmatrix}
u_{1}^{m} \\         u_{2}^{m} \\   \vdots \\     u_{N-1}^{m}
\end{bmatrix},\:\:\:
P^{m}=\frac{\delta x}{3}\begin{bmatrix}
u_{0}^{m}g_{n,0} \\         0 \\     \vdots \\    u_{N}^{m}g_{n,N}
\end{bmatrix}.
\end{equation*}
The matrix $B_{g}$ can be transformed into a Toeplitz matrix by transferring the coefficient $[4,2,4,...,2,4]^{T}$ to the vector $u$ as follows
\begin{center}
	$\tilde{B}_{g}=\frac{\delta x}{3}\left[
	\begin{array}{ccccc}
	g(x_1-x_1) & g(x_2-x_1) & g(x_3-x_1) &  \dots  & g(x_{N-1}-x_1) \\
	g(x_1-x_2) & g(x_2-x_2) & g(x_3-x_2) &  \dots  & g(x_{N-1}-x_2) \\
	g(x_1-x_3) & g(x_2-x_3) & g(x_3-x_3) &  \dots  & g(x_{N-1}-x_3) \\
	\dots &   \dots & \dots & \dots & \dots \\
	g(x_1-x_{N-1}) & g(x_2-x_{N-1}) & g(x_3-x_{N-1}) &  \dots  & g(x_{N-1}-x_{N-1}) \\
	\end{array}
	\right],$
\end{center}
and
\[
\tilde{u}^{m}=\left[4u_{1}^{m},2u_{2}^{m},..,2u_{N-2}^{m},4u_{N-1}^{m}\right]^T.
\]
The above matrix-vector product $(\tilde{B}_{g}\tilde{u}^{m})$ is obtained with $O(N\log{}N)$ complexity by embedding the matrix $\hat{B}_{g}$ in a circulant matrix and using FFT for matrix-vector multiplication \cite{Chan07, Chan96}. Therefore, the discrete approximation $(\mathbb{I}_{\delta}u)$ for the integral operator $(\mathbb{I}u)$ is
\begin{equation}
\mathbb{I_{\delta}}u^{m}=\lambda\left(\tilde{B}_{g}\tilde{u}^{m}+P^{m}+\Upsilon(x,\tau,L)\right).
\end{equation}
If $\mathbb{L}_{\delta}$ denote the discrete approximation of operator $\mathbb{L}$ (defined in Equation~(\ref{eq:l})) then
\begin{equation}
\label{eq:ldelta}
\mathbb{L}_{\delta}u^m_{n}=\mathbb{D_{\delta}}\left(\frac{u^{m+1}_{n}+u^{m-1}_{n}}{2}\right)+\mathbb{I}_{\delta}u^m_{n}-(r+\lambda)u^m_{n}.
\end{equation}
The above three-time levels discretization $(\ref{eq:ldelta})$ of integro-differential operator $\mathbb{L}$ is used for the solution of PIDE (\ref{eq:pidefinal}). We find $U^{m}_{n}$ (the approximate value of $u^{m}_{n}$) which is the solution of following problem
\begin{equation}
\label{eq:pidediscre}
\frac{U^{m+1}_{n}-U^{m-1}_{n}}{2\delta\tau}=\mathbb{D_{\delta}}\left(\frac{U^{m+1}_{n}+U^{m-1}_{n}}{2}\right)+\mathbb{I_{\delta}}U^{m}_{n}-(r+\lambda)U^{m}_{n},\:\:\:\mbox{for $1\leq m\leq M-1$},
\end{equation}
with suitable initial and boundary conditions. Using the value of $D_{\delta}U^{m}_{n}$ from Equation~(\ref{eq:ddelta}) in Equation~(\ref{eq:pidediscre}), we obtain
\begin{equation}
\label{eq:corec2}
\begin{split}
\frac{U^{m+1}_{n}-U^{m-1}_{n}}{2\delta\tau}
&=\frac{1}{2}\left[\frac{\sigma^2}{2}\left(2\Delta_{x}^2U^{m+1}_{n}-\Delta_{x}U_{x_{n}}^{m+1}\right)+\left(r-\frac{\sigma^2}{2}-\lambda \zeta\right)U_{x_{n}}^{m+1}\right]\\
&+\frac{1}{2}\left[\frac{\sigma^2}{2}\left(2\Delta_{x}^2U^{m-1}_{n}-\Delta_{x}U_{x_{n}}^{m-1}\right)+\left(r-\frac{\sigma^2}{2}-\lambda \zeta\right)U_{x_{n}}^{m-1}\right]\\
&+\mathbb{I_{\delta}}U^{m}_{n}-(r+\lambda)U^{m}_{n},\:\:\:\mbox{for $1\leq m\leq M-1$}.
\end{split}
\end{equation}
Re-arranging the terms, we get
\begin{equation}
\label{eq:corec3}
\begin{split}
(I-\delta \tau \frac{\sigma^2}{2} 2 \Delta_{x}^2)U^{m+1}_{n}&={\delta \tau}\left[-\frac{\sigma^2}{2}\Delta_{x}U_{x_{n}}^{m+1}+\left(r-\frac{\sigma^2}{2}-\lambda \zeta\right)U_{x_{n}}^{m+1}\right]\\
&+{\delta\tau}\left[\frac{\sigma^2}{2}\left(2\Delta_{x}^2U^{m-1}_{n}-\Delta_{x}U_{x_{n}}^{m-1}\right)+\left(r-\frac{\sigma^2}{2}-\lambda \zeta\right)U_{x_{n}}^{m-1}\right]\\
&+U^{m-1}_{n}+2\delta \tau \mathbb{I_{\delta}}U^{m}_{n}-2\delta \tau (r+\lambda)U^{m}_{n},\:\:\:\mbox{for $1 \leq m\leq M-1$}.
\end{split}
\end{equation}
Let us introduce the following notation
\[
\textbf{U}^m=(U_{1}^m,U_{2}^m,...,U_{N-1}^m)^T \: \mbox{and} \:\:\:\textbf{U}_{x}^m=(U_{x_{1}}^m,U_{x_{2}}^m,...,U_{x_{N-1}}^m)^T,
\]
the resulting system of equations corresponding to the difference scheme (\ref{eq:corec3}) can be written as
\begin{equation}
\label{eq:corec4}
A\textbf{U}^{m+1}=F(\textbf{U}^{m}, \textbf{U}^{m-1}, \textbf{U}_{x}^{m-1},\textbf{U}_{x}^{m+1}).
\end{equation}
The presence of $\textbf{U}_{x}^{m+1}$ on the right hand side of the Equation~(\ref{eq:corec4}) bind us to use a predictor corrector method. Therefore, correcting to convergence approach is used and also summarized in the following algorithm \cite{Lambert91}.\\
\textbf{Algorithm for Correcting to Convergence Approach}\\
1. Start with $\textbf{U}^{m}$. \\
2. Obtain $\textbf{U}_{x}^{m}$ using Equation~(\ref{eq:firstc_4}). \\
3. Take $\textbf{U}^{m+1}_{old}=\textbf{U}^{m}$, $\textbf{U}_{x_{old}}^{m+1}=\textbf{U}_{x}^{m}$.\\
4. Correct to $\textbf{U}^{m+1}_{new}$ using Equation~(\ref{eq:corec3}).\\
5. If $\|\textbf{U}^{m+1}_{new}-\textbf{U}^{m+1}_{old}\|_{\infty}$ $<$ $\epsilon$, then $\textbf{U}^{m+1}_{new}=\textbf{U}^{m+1}_{old}$.\\
6. Obtain $\textbf{U}_{x_{new}}^{m+1}$ using Equation~(\ref{eq:firstc_4}).\\
7. Take $\textbf{U}^{m+1}_{old}=\textbf{U}^{m+1}_{new}$, $\textbf{U}^{m+1}_{x_{old}}=\textbf{U}^{m+1}_{x_{new}}$ and go to step $4$.\\
The stopping criterion for inner iteration can be set at $\epsilon=10^{-12}$ in above approach. Approximately four iterations are needed at
each time interval to get $\|\textbf{U}^{m+1}_{new}-\textbf{U}^{m+1}_{old}\|_{\infty} < 10^{-12}$. Since the proposed compact scheme~(\ref{eq:corec4}) is three-time levels, two initial values on the zeroth and first time levels are required to start the computation. The initial condition provides the value of $u$ at $\tau =0$ and the value of $u$ at first time level is obtained by IMEX-scheme used in \cite{Cont05}.
\section{Consistency and Stability Analysis}
\label{sec:analysis}
\subsection{Consistency}
\label{ssec:consistency}
\par The consistency of the proposed three-time levels compact scheme ~(\ref{eq:corec3}) is proved in the following theorem.
\begin{theorem}
	Let $u \in C^{\infty} \left([-L,L]\times(0,T]\right)$ satisfy the initial and boundary conditions~(\ref{eq:initial_put})-~(\ref{eq:boundary_put}). Then for sufficiently small $\delta \tau$ and $\delta x$,
	\small
	\begin{equation}
	\label{eq:temp_cons}
	\frac{\partial u}{\partial \tau}(x_n,\tau_m)-\mathbb{L}u(x_n,\tau_m)-\left(\frac{u(x_n,\tau_{m+1})-u(x_n,\tau_{m-1})}{2 \delta \tau}-\mathbb{L}_{\delta}u(x_n,\tau_m)\right)=O(\delta \tau^2+\delta x^4),
	\end{equation}
	\normalsize
	where $\mathbb{L}$ and $\mathbb{L}_{\delta}$ are defined in Equations~(\ref{eq:operator}) and ~(\ref{eq:ldelta}) respectively and $(x_n,\tau_m) \in (-L,L)\times(0,T]$.
\end{theorem}
\begin{proof} The second-order accurate finite difference approximation for $\frac{\partial u}{\partial \tau}$ at each grid point $(x_{n},\tau_{m})$ can be written as
	\small
	\begin{equation}
	\label{eq:timeerr}
	\left|\frac{\partial u}{\partial \tau}(x_n,\tau_m)-\frac{u(x_n,\tau_{m+1})-u(x_n,\tau_{m-1})}{2 \delta \tau}\right|\\
	\leq \frac{\delta \tau^2}{6} \sup_{\tau \in [\tau_{m-1}, \tau_{m+1}]}\left|\frac{\partial^3 u}{\partial \tau^3}(x_n,\tau)\right|.
	\end{equation}
	\normalsize
	Since the compact finite difference approximations for first and second derivatives (discussed in Section~\ref{sec:compact}) are fourth-order accurate, therefore
	\[
	\left\lvert\frac{\partial^2 u}{\partial x^2}(x_n,\tau_{m+1})-u_{xx_{n}}^{m+1}\right\rvert=O(\delta x^4),\:\:\:\:\:\:\left \lvert\frac{\partial^2 u}{\partial x^2}(x_n,\tau_{m-1})-u_{xx_{n}}^{m-1}\right\rvert=O(\delta x^4),
	\]
	\[
	\left\lvert\frac{\partial u}{\partial x}(x_n,\tau_{m+1})-u_{x_{n}}^{m+1}\right\rvert=O(\delta x^4),\:\:\:\:\:\:\left \lvert\frac{\partial u}{\partial x}(x_n,\tau_{m-1})-u_{x_{n}}^{m-1}\right\rvert=O(\delta x^4).
	\]
	Let us now discuss the first and second derivatives in operator $\mathbb{D}u$. From Taylor series expansion for second derivative, we write
	\small
	\[
	\left \lvert \frac{\partial^2 u}{\partial x^2}(x_n,\tau_m)-\frac{1}{2}\left[\frac{\partial^2 u}{\partial x^2}(x_n,\tau_{m+1})+\frac{\partial^2 u}{\partial x^2}(x_n,\tau_{m-1})\right]\right\rvert \leq \frac{\delta \tau^2}{2} \sup_{\tau \in [\tau_{m-1}, \tau_{m+1}]}\left|\frac{\partial^4 u}{\partial x^2 \partial \tau^2 }(x_n,\tau)\right|.
	\]
	\normalsize
	The following relation can be deduced for compact finite difference approximation for second derivative
	\small
	\begin{equation*}
	\begin{aligned}
	\frac{\partial^2 u}{\partial x^2}(x_n,\tau_m)-\frac{1}{2}\left[u_{xx_{n}}^{m+1}+u_{xx_{n}}^{m-1}\right]
	& = \frac{\partial^2 u}{\partial x^2}(x_n,\tau_m)-\frac{1}{2}\left[u_{xx_{n}}^{m+1}+u_{xx_{n}}^{m-1}\right]-\frac{1}{2}\frac{\partial^2 u}{\partial x^2}(x_n,\tau_{m+1})\\
	& + \frac{1}{2}\frac{\partial^2 u}{\partial x^2}(x_n,\tau_{m+1})-\frac{1}{2}\frac{\partial^2 u}{\partial x^2}(x_n,\tau_{m-1})+\frac{1}{2}\frac{\partial^2 u}{\partial x^2}(x_n,\tau_{m-1}),\\
	& = \frac{\partial^2 u}{\partial x^2}(x_n,\tau_m)-\frac{1}{2}\left[\frac{\partial^2 u}{\partial x^2}(x_n,\tau_{m+1})+\frac{\partial^2 u}{\partial x^2}(x_n,\tau_{m-1})\right]\\
	& + \frac{1}{2}\left[\frac{\partial^2 u}{\partial x^2}(x_n,\tau_{m+1})-u_{xx_{n}}^{m+1}\right]+\frac{1}{2}\left[\frac{\partial^2 u}{\partial x^2}(x_n,\tau_{m-1})-u_{xx_{n}}^{m-1}\right],\\
	& = O(\delta \tau^2+\delta x^4).\\
	\end{aligned}
	\end{equation*}
	\normalsize
	From Taylor series expansion for the first derivative, we get
	\small
	\[
	\left\lvert\frac{\partial u}{\partial x}(x_n,\tau_m)-\frac{1}{2}\left[\frac{\partial u}{\partial x}(x_n,\tau_{m+1})+\frac{\partial u}{\partial x}(x_n,\tau_{m-1})\right]\right\rvert \leq \frac{\delta \tau^2}{2} \sup_{\tau \in [\tau_{m-1}, \tau_{m+1}]}\left|\frac{\partial^3 u}{\partial x \partial \tau^2 }(x_n,\tau)\right|.
	\]
	\normalsize
	The compact finite difference approximation for first derivative provides the following relation
	\small
	\begin{equation*}
	\begin{split}
	\frac{\partial u}{\partial x}(x_n,\tau_m)-\frac{1}{2}\left[u_{x_{n}}^{m+1}+u_{x_{n}}^{m-1}\right]
	& = \frac{\partial u}{\partial x}(x_n,\tau_m)-\frac{1}{2}\left[u_{x_{n}}^{m+1}+u_{x_{n}}^{m-1}\right]-\frac{1}{2}\frac{\partial u}{\partial x}(x_n,\tau_{m+1})\\
	& + \frac{1}{2}\frac{\partial u}{\partial x}(x_n,\tau_{m+1})-\frac{1}{2}\frac{\partial u}{\partial x}(x_n,\tau_{m-1})+\frac{1}{2}\frac{\partial u}{\partial x}(x_n,\tau_{m-1}),\\
	& = \frac{\partial u}{\partial x}(x_n,\tau_m)-\frac{1}{2}\left[\frac{\partial u}{\partial x}(x_n,\tau_{m+1})+\frac{\partial u}{\partial x}(x_n,\tau_{m-1})\right]\\
	& + \frac{1}{2}\left[\frac{\partial u}{\partial x}(x_n,\tau_{m+1})-u_{x_{n}}^{m+1}\right]+\frac{1}{2}\left[\frac{\partial u}{\partial x}(x_n,\tau_{m-1})-u_{x_{n}}^{m-1}\right],\\
	& = O(\delta \tau^2+\delta x^4).
	\end{split}
	\end{equation*}
	\normalsize
	Therefore, the error between the operators $\mathbb{D}_{\delta}u$ and $\mathbb{D}u$ is
	\begin{equation}
	\label{eq:differror}
	\mathbb{D} u(x_n,\tau_m)-\mathbb{D}_{\delta}\left(\frac{u(x_n,\tau_{m+1})+u(x_n,\tau_{m-1})}{2}\right)=O(\delta \tau^2+\delta x^4).
	\end{equation}
	Further, Equation~(\ref{eq:simp}) provides the error between the integral operator $\mathbb{I}u$ and $\mathbb{I}_{\delta}u$  as 
	\begin{equation}
	\label{eq:interror}
	\mathbb{I}u(x_n,\tau_m)-\mathbb{I}_{\delta}u(x_n,\tau_m)=O(\delta x^4).
	\end{equation}
	From Equations~(\ref{eq:timeerr}),
	~(\ref{eq:differror}) and~(\ref{eq:interror}), result follows.
\end{proof}
\subsection{Stability}
\label{ssec:stability}
The stability of the proposed compact scheme is proved using von-Neumann stability analysis. Consider a single node
\begin{equation}
\label{eq:u_fou}
U_{n}^{m}=p^{m}e^{in\theta},
\end{equation}
where $i=\sqrt{-1}$, $p^{m}$ is the $m^{th}$ power of amplitude at time levels $\tau_{m}$, and $\theta= 2\pi/N$. The integral operator~(\ref{eq:all_operator}) can be re-written in an equivalent form as 
\[
\mathbb{I}u(x,\tau) = \lambda \int_{-L}^{L}u(y+x,\tau)g(y)dy.
\]
Fourth-order accurate composite Simpson's rule for above equation provides
\begin{equation*}
\begin{split}
\mathbb{I}_{\delta}u&=\delta x \sum_{k=0}^{N}w_{k}U_{k+n}^{m}g_{k},\\
&= \delta x \sum_{k=0}^{N}w_{k}p^{m}e^{i\theta(k+n)}g_{k},\\
&\equiv p^{m}e^{i\theta n}G_{k},
\end{split}
\end{equation*}
where
\begin{equation}
\label{eq:g_k}
G_{k}=\delta x\sum_{k=0}^{N}w_{k}e^{i\theta k}g_{k} \:\:\: \mbox{and} \:\:\: g_{k}=g(x_{k}).
\end{equation}
The following Lemma is proved for numerical quadrature $G_{k}$ given in Equation~(\ref{eq:g_k}).
\begin{lemma}
	\label{lemma:1}
	The numerical quadrature $G_{k}$  satisfies the following
	\[
	\lvert G_{k}\rvert \leq 1+ c \delta x^4,
	\]
	where $c$ is a constant.
\end{lemma}
\begin{proof}
	Using the property of a density function we can write
	\begin{equation}\label{eq:lemma11}
	\int_{\Omega}g(x)dx \leq \int_{-\infty}^{\infty}g(x)dx=1.
	\end{equation}
	The application of fourth-order accurate composite Simpson's rule in the above Equation~(\ref{eq:lemma11}) gives
	\begin{equation}\label{eq:lemma22}
	\delta x \sum_{k=0}^{N}w_{k}g_{k} \leq 1+ c \delta x^4.
	\end{equation}
	From Equations~(\ref{eq:g_k}) and~(\ref{eq:lemma22}), we have
	\begin{equation}
	\begin{split}
	\lvert G_{k}\rvert &= \lvert \delta x \sum_{k=0}^{N}w_{k}e^{i \theta k}g_{k} \rvert,\\
	&  \leq 1+ c \delta x^4.
	\end{split}
	\end{equation}
\end{proof}
For the sake of simplicity, we denote $\frac{\sigma^2}{2}=a$ and $\left(r-\frac{\sigma^2}{2}-\lambda \zeta\right)=b$ in the rest of this section. The fully discrete problem~(\ref{eq:corec3}) can be written in terms of $a$ and $b$ as follows
\begin{equation}
\label{eq:stab}
\begin{split}
(I-2 a \delta \tau \Delta_{x}^2)U^{m+1}_{n}&=(I+2 a \delta \tau 2 \Delta_{x}^2)U^{m-1}_{n}+2 \delta \tau\left[\frac{b}{2}-\frac{a}{2}\Delta_{x}\right]U_{x_{n}}^{m+1}\\
&+2 \delta \tau\left[\frac{b}{2}-\frac{a}{2}\Delta_{x}\right]U_{x_{n}}^{m-1}-2 \delta \tau(r+\lambda)U_{n}^{m}+2 \delta\tau\lambda G_{k}U^{m}_{n}.
\end{split}
\end{equation}
The following relations are obtained from \cite{KPMD17} in order to prove the stability of the proposed compact scheme~(\ref{eq:stab})
\begin{equation}
\begin{split}
\label{eq:stab1}
\Delta_{x}U_{n}^{m}&=i \frac{sin(\theta)}{\delta x}U^{m}_{n},\\
\Delta_{x}^{2}U_{n}^{m}&= \frac{2cos(\theta)-2}{\delta x^2}U^{m}_{n},\\
U_{x_{n}}^{m}&= i \frac{3sin(\theta)}{\delta x(2+cos(\theta))}U^{m}_{n}.
\end{split}
\end{equation}
Using Equation~(\ref{eq:stab1}) in the difference scheme~(\ref{eq:stab}), we get
\small
\begin{equation}
\begin{split}
\label{eq:stab4}
\left[1-4a\delta \tau \left(\frac{cos(\theta)-1}{\delta x^2}\right)\right]U^{m+1}_{n}
&=\left[1+4a\delta \tau \left(\frac{cos(\theta)-1}{\delta x^2}\right)\right]U^{m-1}_{n}+\delta \tau\left[\left(a\frac{sin(\theta)}{\delta x}+ib\right)\right.\\
&\left.\frac{3sin(\theta)}{\delta x(2+cos(\theta))}\right]U_{n}^{m+1}+\delta \tau\left[\left(a\frac{sin(\theta)}{\delta x}+ib\right)\right.\\
&\left.\frac{3sin(\theta)}{\delta x(2+cos(\theta))}\right]U_{n}^{m-1}-2 \delta \tau (r+\lambda)U_{n}^{m}+2\delta \tau \lambda G_{k} U_{n}^{m},
\end{split}
\end{equation}
\normalsize
After re-arranging the terms, the above Equation~(\ref{eq:stab4}) is written as
\small
\begin{equation}
\begin{split}
\label{eq:stab5}
\left[1-\delta \tau \left(a\frac{cos^{2}(\theta)+4cos(\theta)-5}{\delta x^2(2+cos(\theta))}+i b\frac{3sin(\theta)}{\delta x(2+cos(\theta))}\right)\right]U^{m+1}_{n}
&=\left[1+\delta \tau a\frac{cos^{2}(\theta)+4cos(\theta)-5}{\delta x^2(2+cos(\theta))}\right.\\
&\left.+i\delta \tau b\frac{3sin(\theta)}{\delta x(2+cos(\theta))}\right]U^{m-1}_{n}\\
&-2\delta \tau(r+\lambda) U_{n}^{m}+2\delta \tau \lambda G_{k} U_{n}^{m}.
\end{split}
\end{equation}
\normalsize
Using Equation~(\ref{eq:u_fou}) in above Equation~(\ref{eq:stab5}), the amplification polynomial $\Theta(\delta x, \delta \tau, \theta)$ can be written as
\begin{equation}
\label{eq:amplif}
\Theta(\delta x, \delta \tau, \theta)=\gamma_{0}p^2-2\gamma_{1}p-\gamma_{2},
\end{equation}
where
\begin{equation}
\label{eq:gammas}
\begin{split}
\gamma_{0}
& = \left[1-\delta \tau
\left(a\frac{cos^{2}(\theta)+4cos(\theta)-5}{\delta x^2(2+cos(\theta))}+ib\frac{3sin(\theta)}{\delta x(2+cos(\theta))}\right)\right],\\
\gamma_{1}
& = \left[\lambda \delta \tau G_{k}-\delta \tau(r+\lambda)\right],\\
\gamma_{2}
& = \left[1+\delta \tau
\left(a\frac{cos^{2}(\theta)+4cos(\theta)-5}{\delta x^2(2+cos(\theta))}+ib\frac{3sin(\theta)}{\delta x(2+cos(\theta))}\right)\right].\\
\end{split}
\end{equation}
The following lemma is needed to prove the stability of a three-time levels difference scheme, see \cite{Jcstrik04}.
\begin{lemma}
	\label{lemma:2}
	A finite difference scheme is stable if and only if all the roots, $p_{u}$, of the amplification polynomial $\Theta$ satisfies the following condition:\\
	$1$. There is a constant $C$ such that $|p_{u}| \leq
	1+C\delta \tau$.\\
	$2$. There are positive constants $a_{0}$ and $a_{1}$ such that if $a_{0} <|p_{u}|\leq 1+C\delta \tau$ then $|p_{u}|$ is simple root and for any other root $p_{v}$, following relation holds
	\[
	\left|p_{v}-p_{u}\right|\geq a_{1},
	\]
	as $\delta x$, $\delta \tau$$\rightarrow0$.
\end{lemma}
\begin{proof}
	For the proof of above Lemma, see \cite{Jcstrik04}.
\end{proof}
Now, we prove the above Lemma~\ref{lemma:2} for the proposed three-time levels compact scheme for the PIDE in the following theorem.
\begin{theorem}
	The fully discrete problem~(\ref{eq:corec3}) is stable in the sense of von-Neumann for $\delta \tau$ $\leq$ $1/(4\lambda+2r)$.
\end{theorem}
\begin{proof}
	Firstly, some properties of the coefficients $\gamma_{0}, \gamma_{1}$ and $\gamma_{2}$ of amplification polynomial $\Theta(\delta x, \delta \tau, \theta)$ are proved.
	Using Lemma~\ref{lemma:1} in Equation~(\ref{eq:gammas}), it is observed that
	\[
	\lvert \gamma_{1} \rvert < \delta \tau(2\lambda+r).
	\]
	Further, the coefficient $\gamma_{0}$ from Equation~(\ref{eq:gammas}) can be written as
	\begin{equation}
	\label{eq:stab8}
	\lvert \gamma_{0} \rvert=\lvert (1-A)-iB \rvert,
	\end{equation}
	where
	\[
	A=a\frac{cos^{2}(\theta)+4cos(\theta)-5}{\delta x^2(2+cos(\theta))},\:\:\: \mbox{and} \:\:\:B=b\frac{3sin(\theta)}{\delta x(2+cos(\theta))}.
	\]
	Since $a > 0$ $\implies$ $A < 0$, which gives $\lvert \gamma_{0} \rvert > 1 $. Moreover, from Equations~(\ref{eq:gammas})
	and~(\ref{eq:stab8}), we have
	\[
	\left \lvert \frac{\gamma_{2}}{\gamma_{0}} \right \rvert^{2}=\frac{1+A^2+2A+B^2}{1+A^2-2A+B^2},
	\]
	which implies $\left \lvert \frac{\gamma_{2}} {\gamma_{0}} \right \rvert < 1$. Now, roots of the amplification polynomial $\Theta(\delta x, \delta \tau, \theta)$ are
	\begin{equation}
	\begin{split}
	\lvert p\rvert & = \left \lvert\frac{\gamma_{1}\pm \sqrt{\gamma_{1}^2-\gamma_{0}\gamma_{2}}}{\gamma_{0}}\right \rvert, \\
	& \leq \left \lvert\frac{\gamma_{2}}{\gamma_{0}}\right \rvert^{\frac{1}{2}}+ 2\left \lvert\frac{\gamma_{1}}{\gamma_{0}}\right\rvert,\\
	& \leq 1+2\delta \tau (r+2\lambda).
	\end{split}
	\end{equation}
	Hence, first part of the Lemma~\ref{lemma:2} is proved for constant $C=2(r+2\lambda)$. Let us assume
	that $p_{1}$ and $p_{2}$ are two roots of amplification polynomial $\Theta(\delta x, \delta \tau, \theta)$ and the constant $a_{0}=1$ which implies $p_{1}>1$, then
	\begin{equation}
	\begin{split}
	|p_{1}-p_{2}| & \geq 2|p_{1}|-|p_{1}+p_{2}|,\\
	& \geq 2-2\delta \tau(2\lambda+r).
	\end{split}
	\end{equation}
	If $\delta \tau$ satisfies the given condition, we have
	\[
	|p_{1}-p_{2}|\geq 1,
	\]
	and this prove the second part of the Lemma~\ref{lemma:2} with $a_{1}=1$. This completes the proof.
\end{proof}
\section{Numerical Results}
\label{sec:numerical}
\par In this section, the applicability of the proposed compact scheme for pricing European options under jump-diffusion models is demonstrated. According to \cite{thomee70}, fourth-order convergence cannot be expected for non-smooth initial conditions. Since the initial conditions given in Equations~(\ref{eq:initial_call}) and~(\ref{eq:initial_put}) have low regularity, therefore suitable smoothing operator is required to smoothen the initial conditions. For this purpose, the smoothing operator $\phi_4$ given in \cite{thomee70} is employed to smoothen the initial conditions and it's Fourier transform is define as
\[
\hat{\phi}_{4}(\omega)=\left(\frac{sin(\omega/2)}{\omega/2}\right)^4\left[1+\frac{2}{3}sin^2(\omega/2)\right].
\]
As a result, the following smoothed initial condition $(\tilde{u}_{0})$ is obtained
\begin{equation}
\label{eq:smoothed}
\tilde{u}_{0}(x_1)=\frac{1}{\delta x}\int_{-3\delta x}^{3\delta x}\phi_{4}\left(\frac{x}{\delta x}\right)u_{0}(x_1-x)dx,
\end{equation}
where $u_{0}$ is the actual non-smooth initial condition and $x_{1}$ is the grid point where smoothing is required. The smoothed initial conditions obtained from Equation~(\ref{eq:smoothed}) tends to the original initial conditions as $\delta x\rightarrow 0$. The parameters considered for pricing European options under Merton and Kou jump-diffusion models are listed in Table~\ref{table:parameter}. The parabolic mesh ratio $(\frac{\delta \tau}{\delta x^2})$ is fixed as $0.4$ in all our computations, although neither the von Neumann stability analysis nor the numerical experiments showed any such restriction. The $\ell^2$ error $\|U(\delta x, \delta \tau)-U(\frac{\delta x}{2}, \frac{\delta \tau}{2})\|_{\ell^2}$ is used to examine the numerical convergence rate of the proposed compact scheme, where $U(\delta x, \delta \tau)$ represents the solution at the step sizes $\delta x$ and $\delta \tau$ and $U\left(\frac{\delta x}{2}, \frac{\delta \tau}{2}\right)$ represents the solution after halving the step sizes.
\begin{table}[h!]
	\begin{center}
		\begin{tabular}{ P{0.5 cm} P{4.5 cm} | P{4 cm} P{1 cm} }
			\hline
			& Merton jump-diffusion  model & Kou jump-diffusion  model & \\
			\hline
			Parameters & Values & Parameters & Values  \\
			\hline
			$\lambda$ & $0.10$ & $\lambda$ & $0.10$  \\
			$T$ & $0.25$ & $T$ & $0.25$ \\
			$r$ & $0.05$ & $r$ & $0.05$  \\
			$K$ & $100$ & $K$ & $100$  \\
			$\sigma$ & $0.15$ & $\sigma$ & $0.15$  \\
			$\mu_{J}$ & $-0.90$ & $\lambda_{+}$ & $3.0465$  \\
			$\sigma_{J}$ & $0.45$ & $\lambda_{-}$ & $3.0775$  \\
			& & $p$ & $0.3445$  \\
			\hline
		\end{tabular}
	\end{center}
	\caption{The values of parameters for pricing European options under Merton and Kou jump-diffusion models.}
	\label{table:parameter}
\end{table}
\begin{table}[h!]
	\begin{tabular}{ m{4cm} | m{2.3cm} | m{2.3cm}| m{2.3cm} }
		\hline
		& S=90 & S=100 & S=110 \\
		\hline
		Reference values & 9.285418 & 3.149025 & 1.401185 \\
		\hline
		Proposed compact scheme & 9.285420 & 3.149114 & 1.401176 \\
		\hline
	\end{tabular}
	\caption{Values of European put options using proposed three-time levels compact scheme under Merton jump-diffusion model with $N=1536$ for different stock prices.}
	\label{table:merton}
\end{table}
\begin{figure}[h!]
	\begin{center}
		\subfigure[]{%
			\includegraphics[scale=0.420]{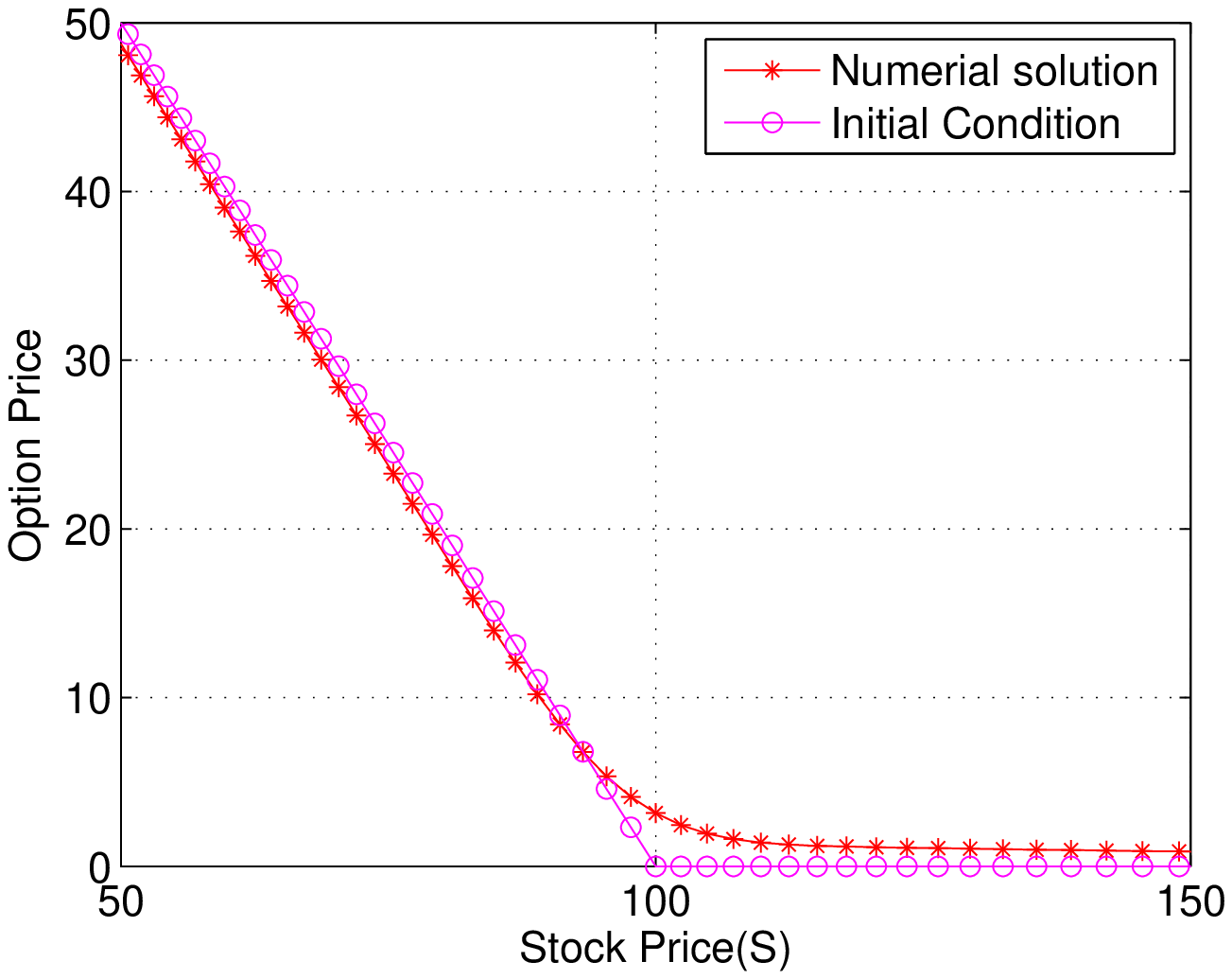}
			\label{fig:price_merton}}%
		\subfigure[]{%
			\includegraphics[scale=0.420]{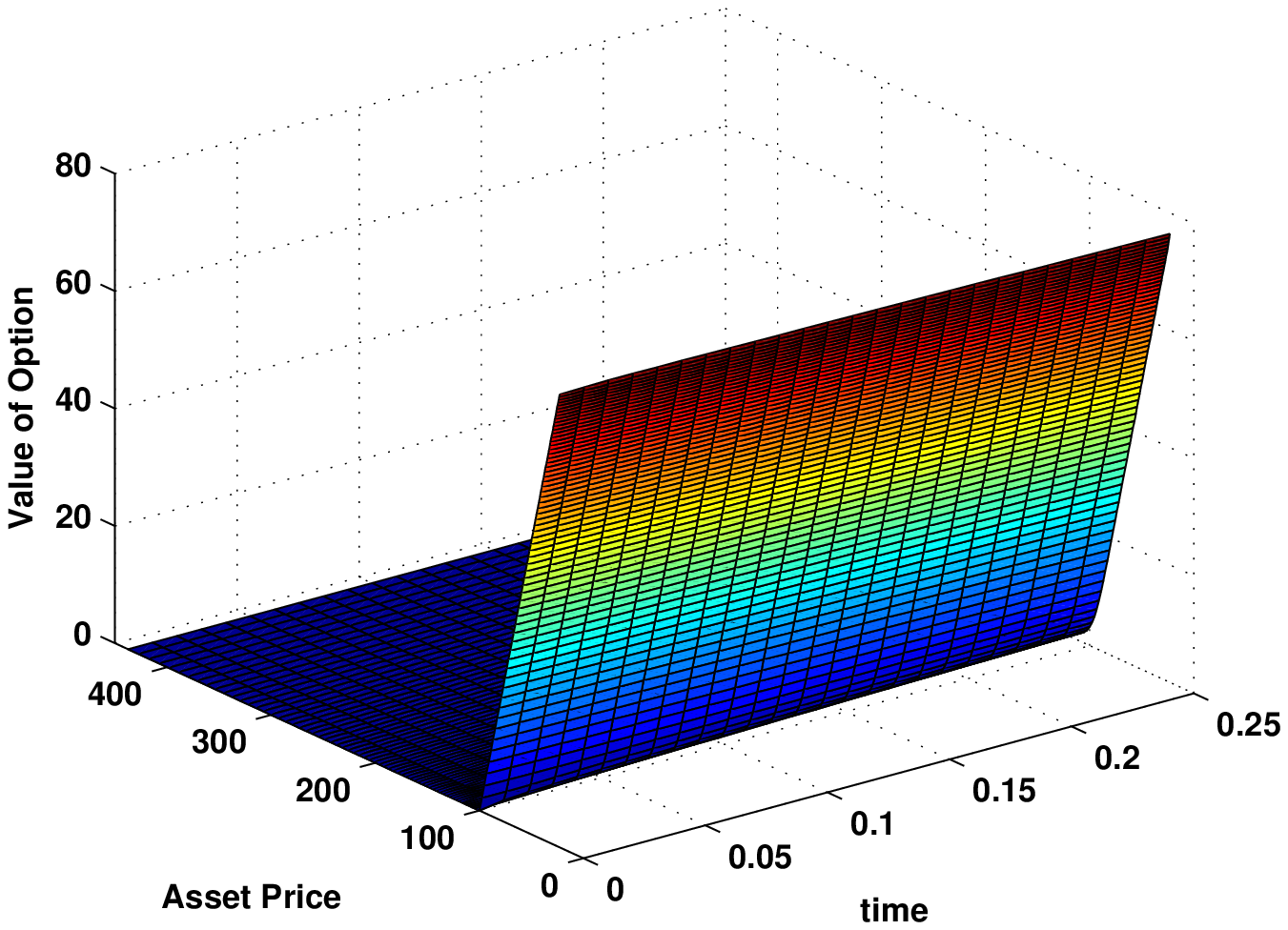}
			\label{fig:merton_price3d}}
		\caption{Values of European put options using proposed three-time levels compact scheme under Merton jump-diffusion model: (a) As a function of stock price, (b) As a function of stock price and time.}
		\label{fig:merton}
	\end{center}
\end{figure}
\begin{figure}[h!]
	\centering
	\includegraphics[width=8 cm]{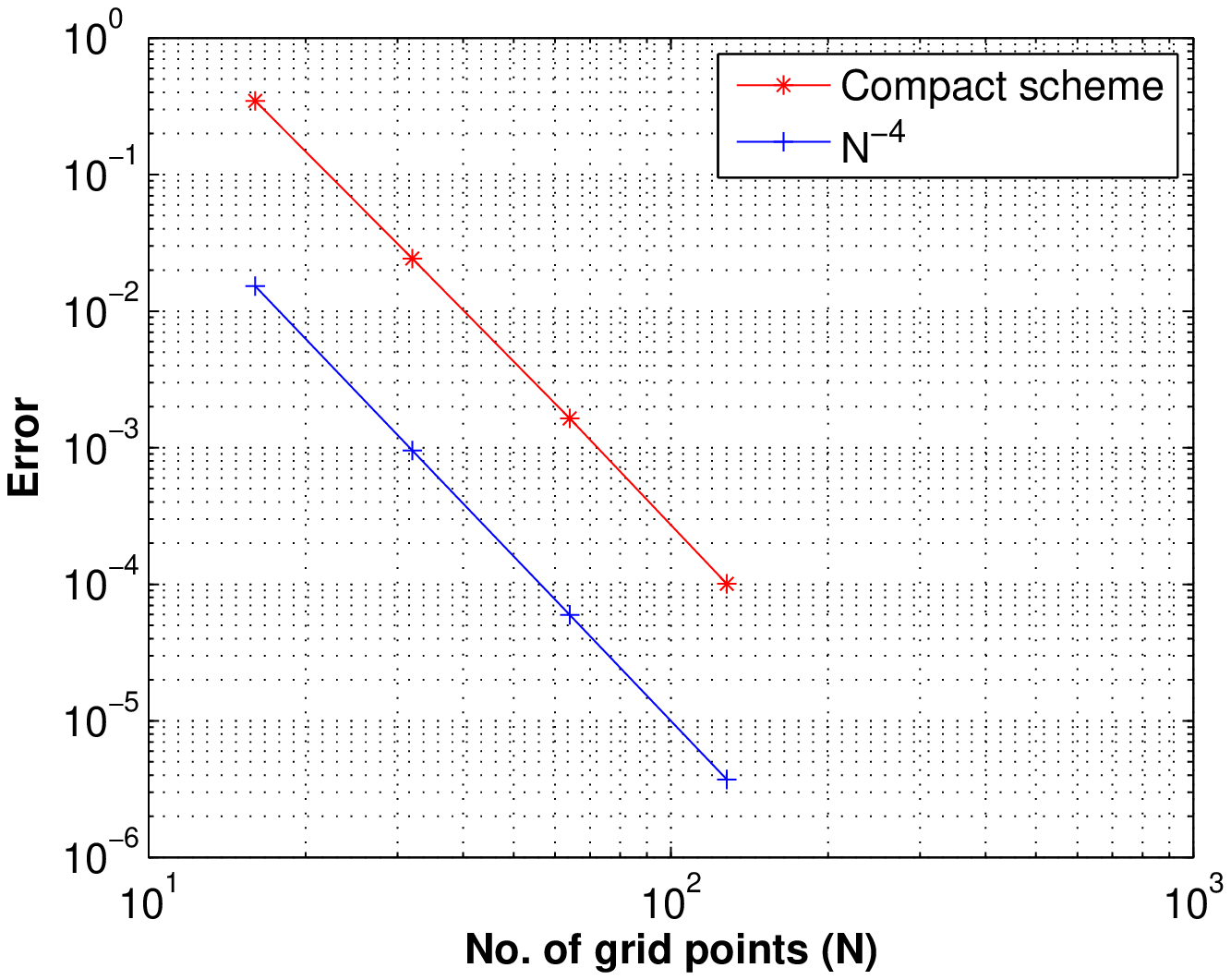}
	\caption{Rate of convergence: Error between the numerical solutions versus number of grid points for European put option under Merton jump-diffusion model.}
	\label{fig:roc_mertonp}
\end{figure}
\begin{figure}[h!]
	\begin{center}
		\subfigure[]{%
			\includegraphics[scale=0.420]{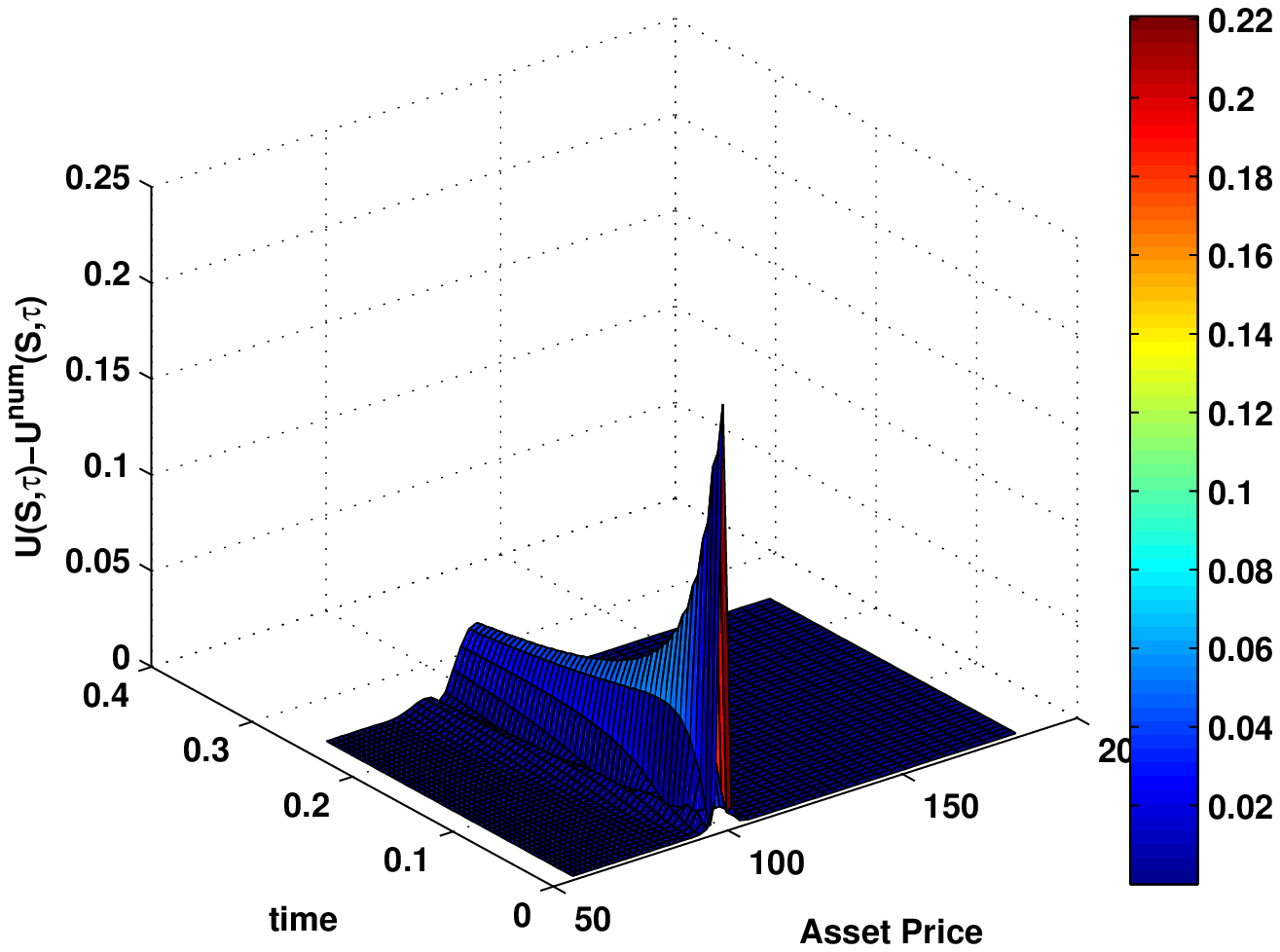}
			\label{fig:err_woutsmo_finite}}%
		\subfigure[]{%
			\includegraphics[scale=0.420]{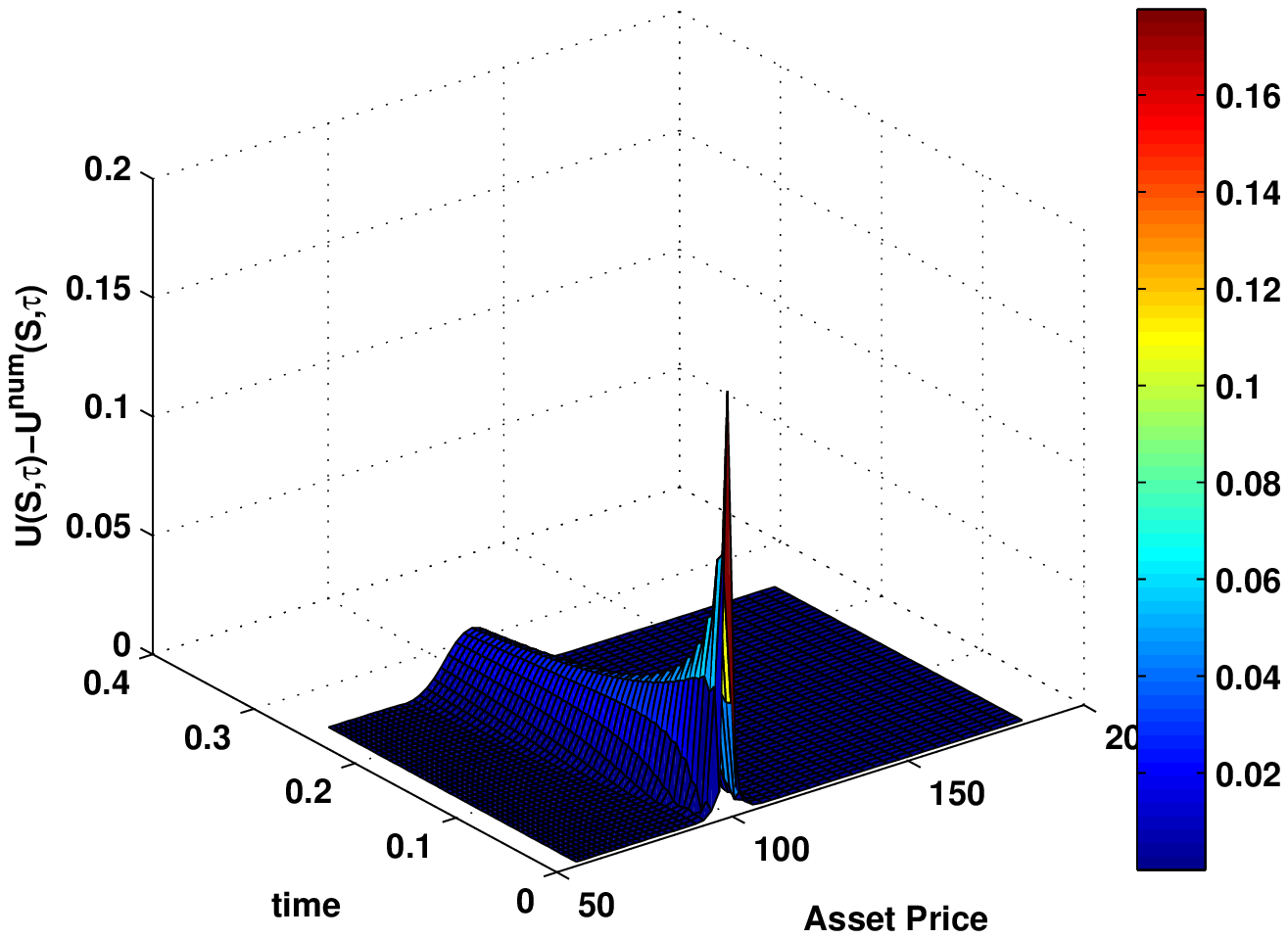}
			\label{fig:err_woutsmo_compact}}
		\subfigure[]{%
			\includegraphics[scale=0.420]{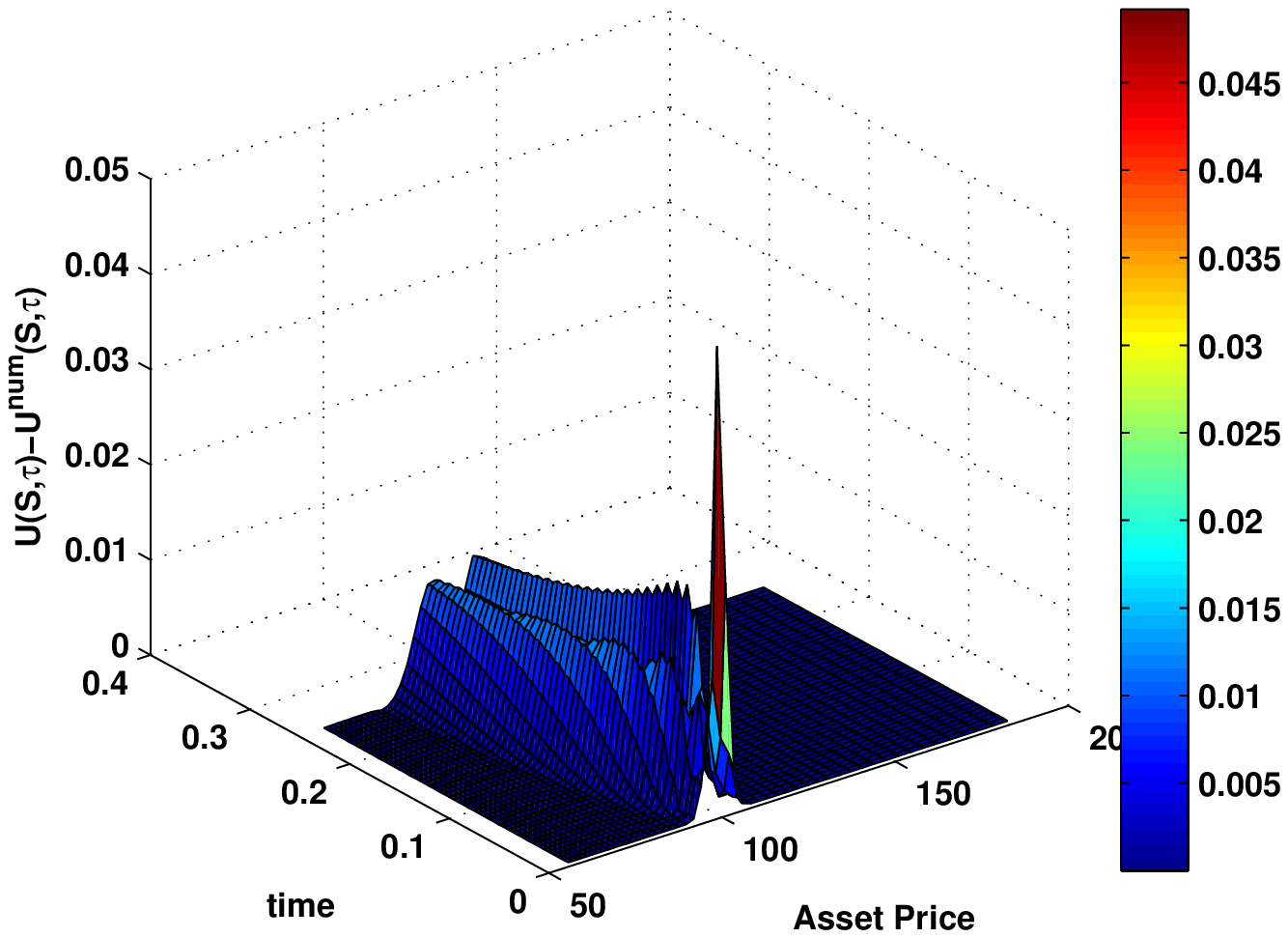}
			\label{fig:err_smo_finite}}%
		\subfigure[]{%
			\includegraphics[scale=0.420]{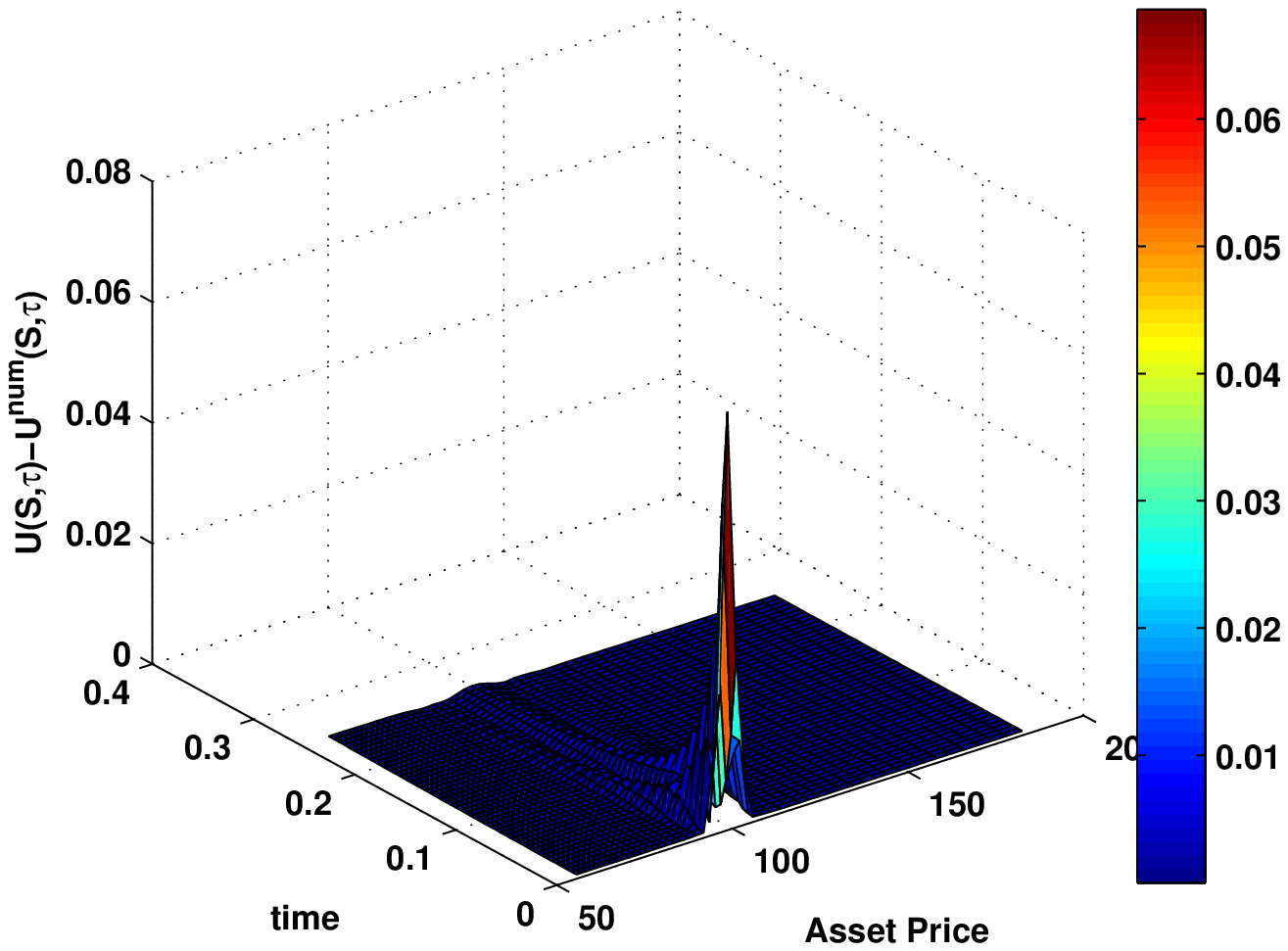}
			\label{fig:err_smo_compact}}
		\caption{The difference between the numerical and analytic solution for European put options under Merton jump-diffusion model as a
			function of stock price and time: (a) Using finite difference scheme with non-smooth initial condition, (b) Using proposed compact scheme with non-smoothing initial condition, (c) Using finite difference scheme with smoothed initial condition, and (d) Using proposed compact scheme with smoothed initial condition.}
		\label{fig:err_smo}
	\end{center}
\end{figure}
\begin{figure}[h!]
	\centering
	\includegraphics[width=8 cm]{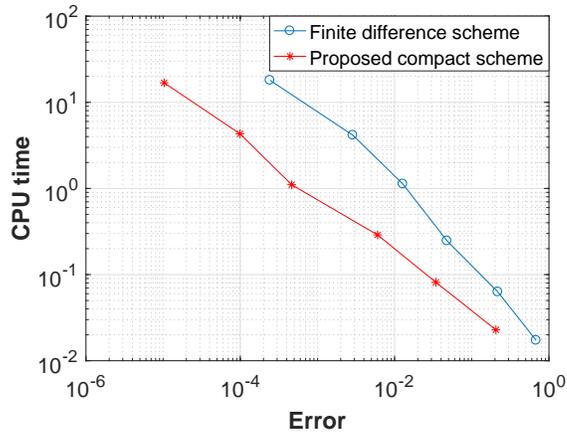}
	\caption{Efficiency: CPU time in seconds versus error using finite difference scheme and proposed compact scheme for pricing European put options under Merton jump-diffusion model.}
	\label{fig:compu}
\end{figure}
\begin{table}[h!]
	\begin{tabular}{ m{4cm} | m{2.3cm} | m{2.3cm}| m{2.3cm} }
		\hline
		& S=90 & S=100 & S=110 \\
		\hline
		Reference values & 0.527638 & 4.391246 & 12.643406 \\
		\hline
		Proposed compact scheme & 0.527636 & 4.391244 & 12.643408 \\
		\hline
	\end{tabular}
	\caption{Values of European call options using proposed three-time levels compact scheme under Merton jump-diffusion model with $N=1536$ for different stock prices.}
	\label{table:mertonc}
\end{table}
\begin{figure}[h!]
	\centering
	\includegraphics[width=8 cm]{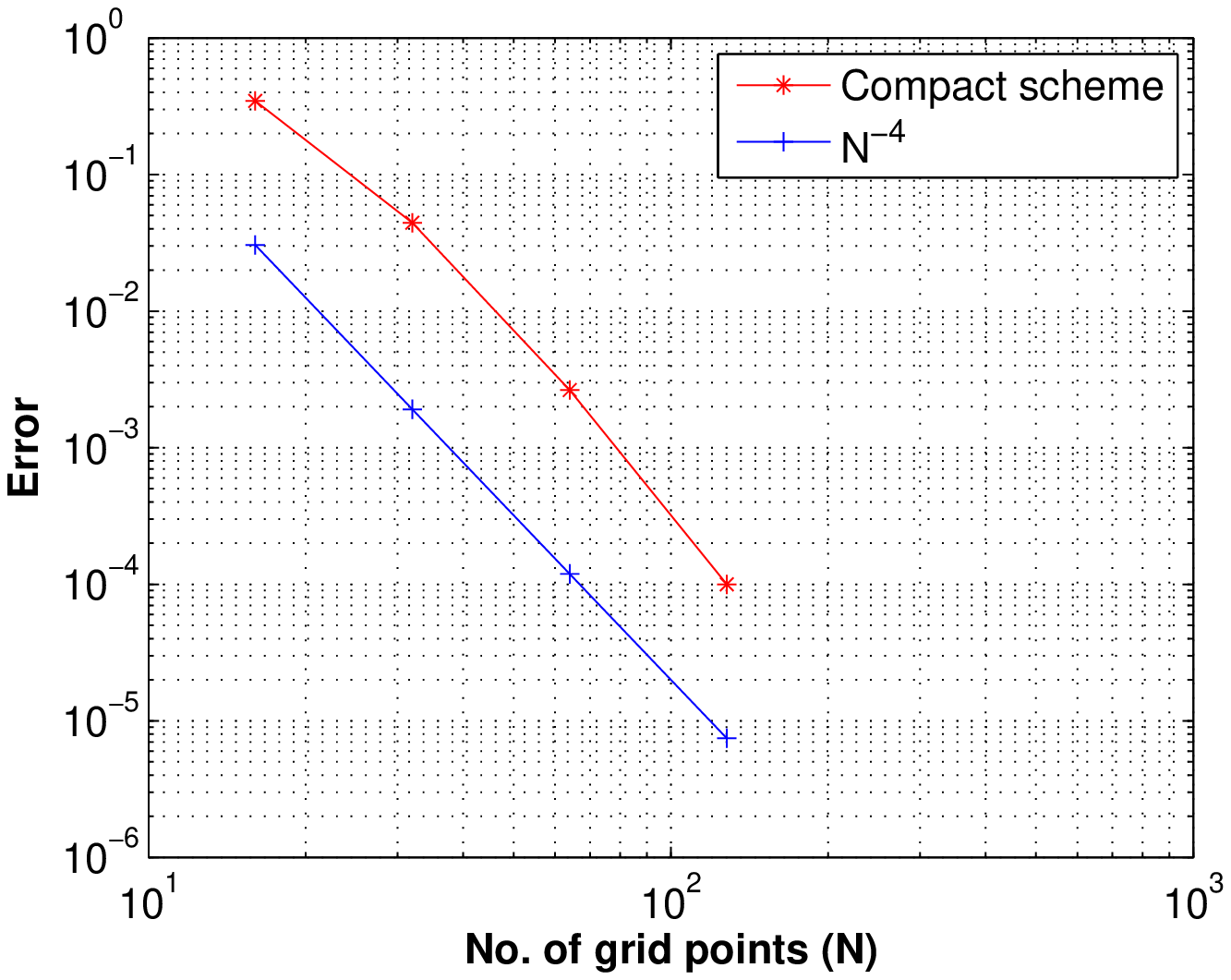}
	\caption{Rate of convergence: Error between the numerical solutions versus number of grid points for European call option under Merton jump-diffusion model.}
	\label{fig:roc_mertonc}
\end{figure}
\begin{figure}[h!]
	\begin{center}
		\subfigure[]{%
			\includegraphics[scale=0.420]{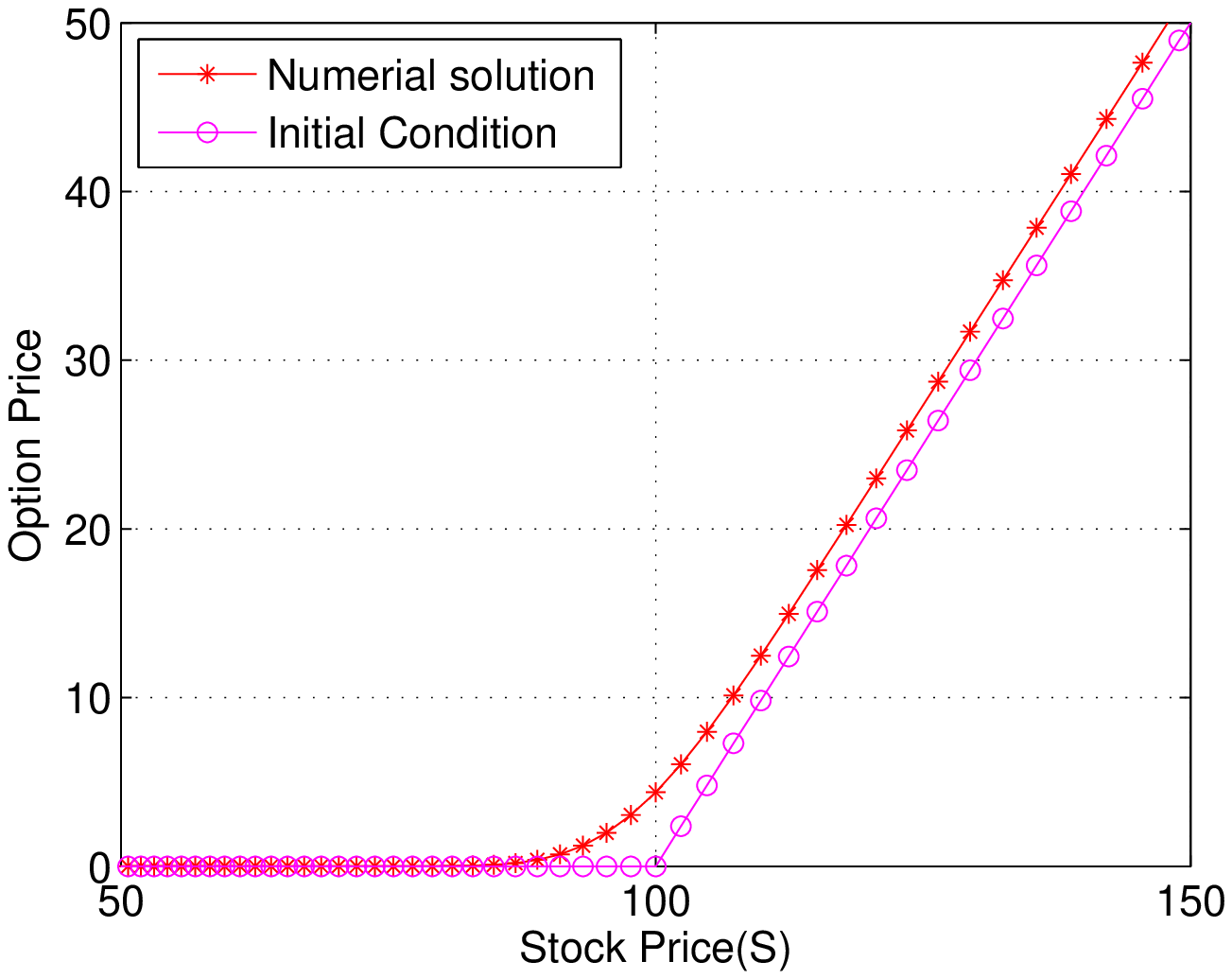}
			\label{fig:price_mertonc}}%
		\subfigure[]{%
			\includegraphics[scale=0.420]{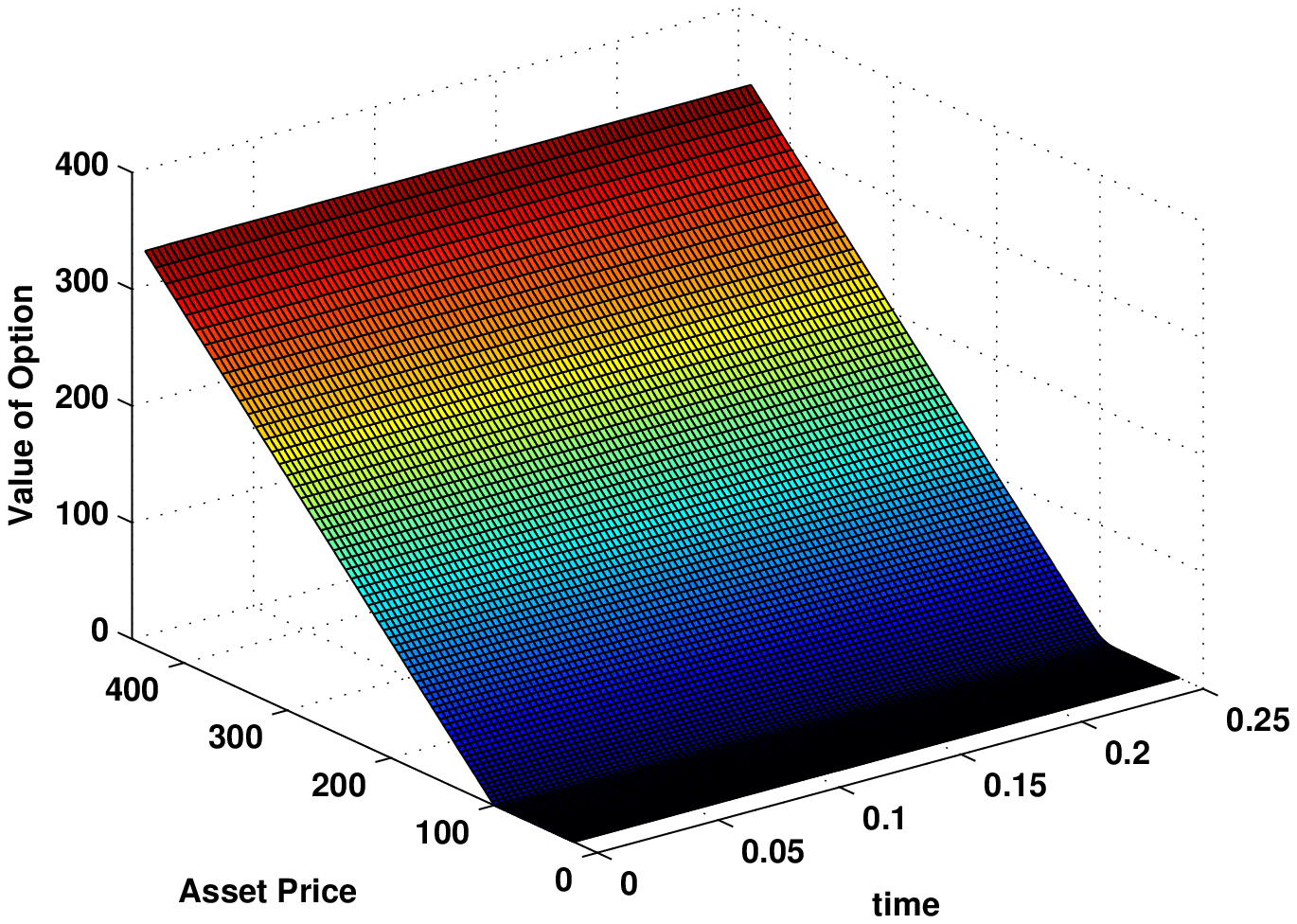}
			\label{fig:mertonc_price3d}}
		\caption{Values of European call options using proposed three-time levels compact scheme under Merton jump-diffusion model: (a) As a function of stock price, (b) As a function of stock price and time.}
		\label{fig:mertonc}
	\end{center}
\end{figure}
\begin{table}[h!]
	\begin{tabular}{ m{4cm} | m{2.3cm} | m{2.3cm}| m{2.3cm} }
		\hline
		& S=90 & S=100 & S=110 \\
		\hline
		Reference values & 9.430457 & 2.731259 & 0.552363 \\
		\hline
		Proposed compact scheme & 9.430448 & 2.731252 & 0.552361 \\
		\hline
	\end{tabular}
	\caption{Values of European put options using proposed three-time levels compact scheme under Kou jump-diffusion model with $N=1536$ for different stock prices.}
	\label{table:kou}
\end{table}
\begin{figure}[h!]
	\begin{center}
		\subfigure[]{%
			\includegraphics[scale=0.420]{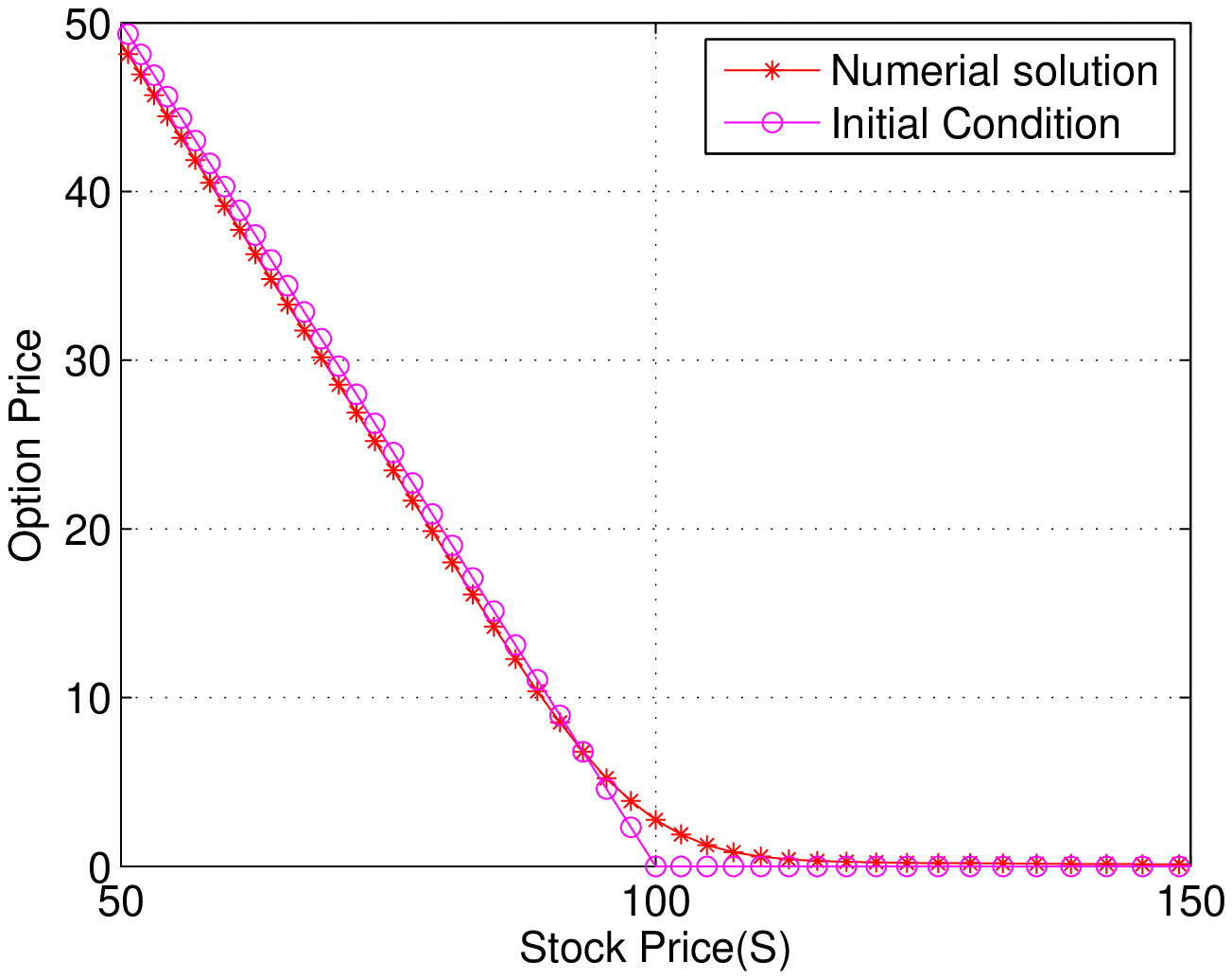}
			\label{fig:price_kou}}%
		\subfigure[]{%
			\includegraphics[scale=0.420]{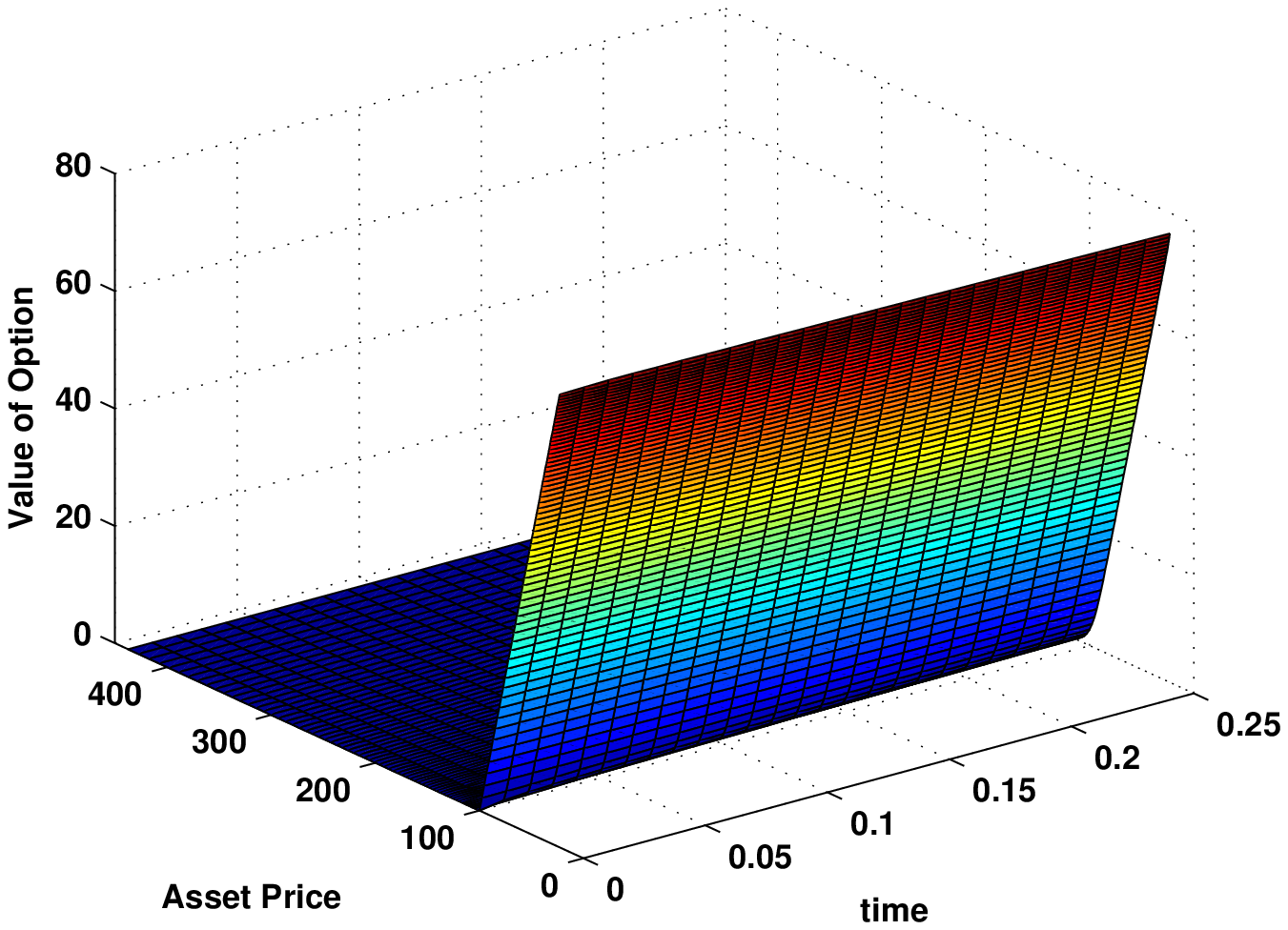}
			\label{fig:kou_price3d}}
		\caption{Values of European put options using proposed three-time levels compact scheme under Kou jump-diffusion model: (a) As a function of stock price, (b) As a function of stock price and time.}
		\label{fig:kou}
	\end{center}
\end{figure}
\begin{figure}[h!]
	\centering
	\includegraphics[width=8 cm]{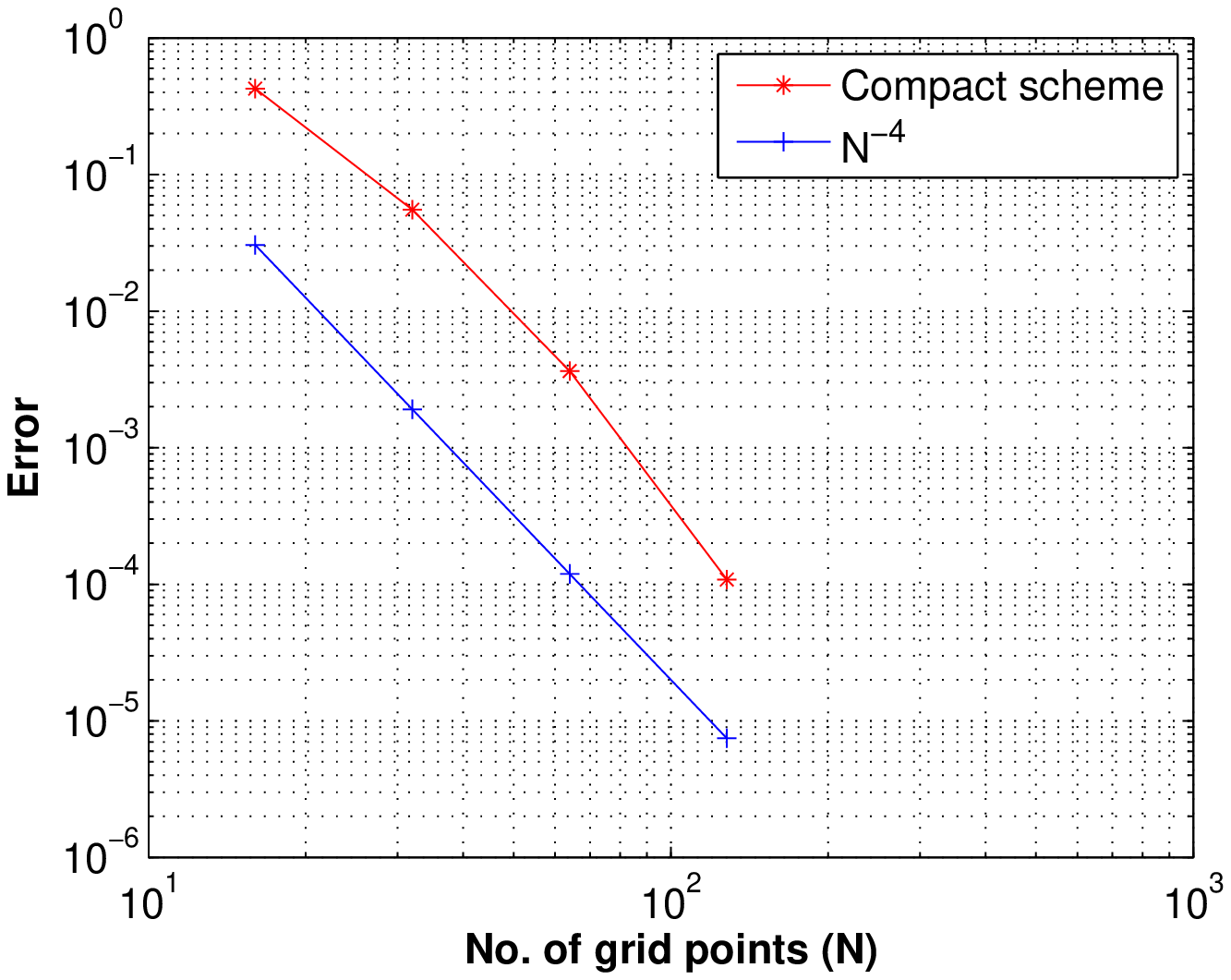}
	\caption{Rate of convergence: Error between the numerical solutions versus number of grid points for European put option under Kou jump-diffusion model.}
	\label{fig:roc_koup}
\end{figure}
\begin{table}[h!]
	\begin{tabular}{ m{4cm} | m{2.3cm} | m{2.3cm}| m{2.3cm} }
		\hline
		& S=90 & S=100 & S=110 \\
		\hline
		Reference values & 0.672677 & 3.973479 & 11.794583 \\
		\hline
		Proposed compact scheme & 0.672672 & 3.973476 & 11.794584 \\
		\hline
	\end{tabular}
	\caption{Values of European call options using proposed three-time levels compact scheme under Kou jump-diffusion model with $N=1536$ for different stock prices.}
	\label{table:kouc}
\end{table}
\begin{figure}[h!]
	\begin{center}
		\subfigure[]{%
			\includegraphics[scale=0.420]{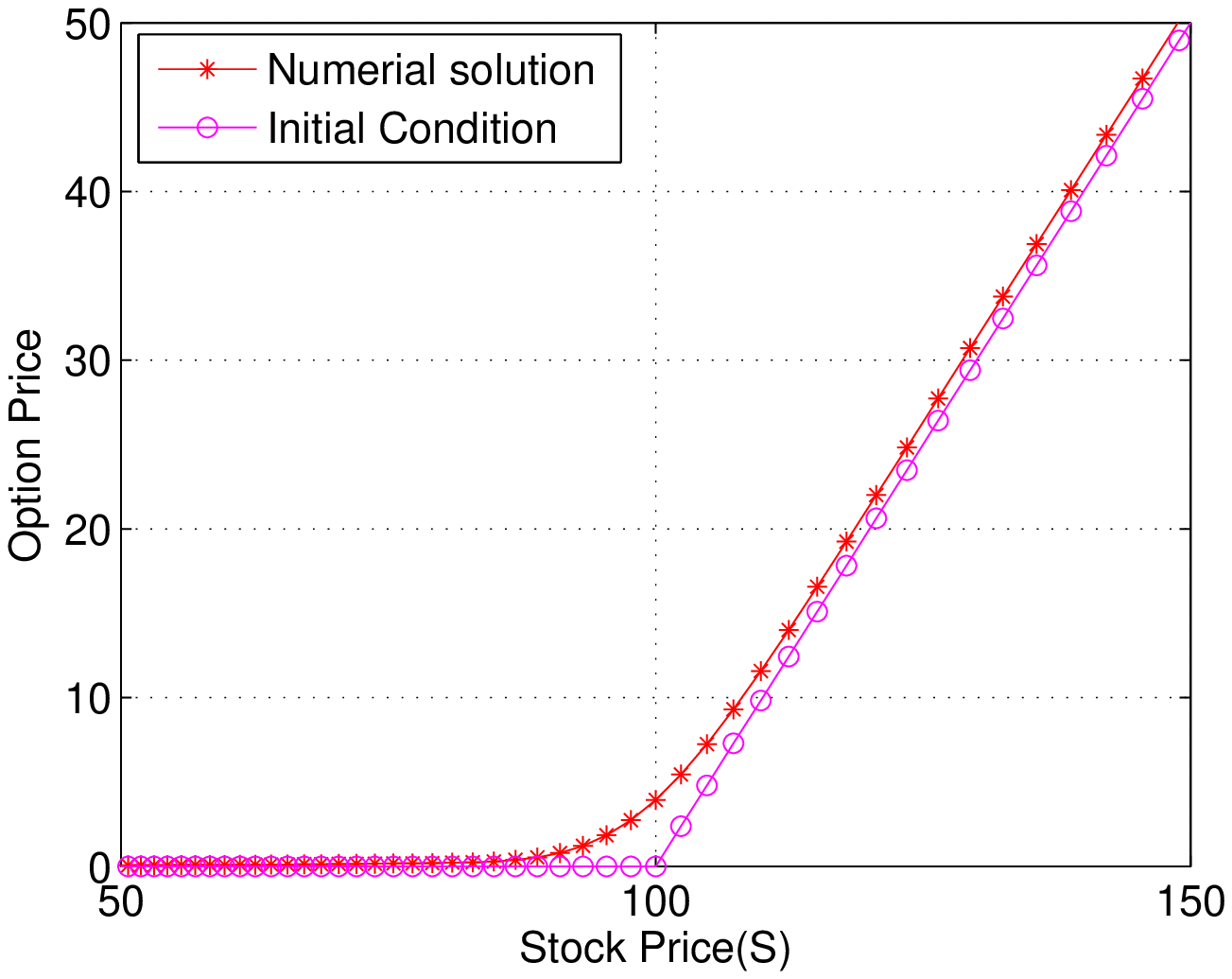}
			\label{fig:price_kouc}}%
		\subfigure[]{%
			\includegraphics[scale=0.420]{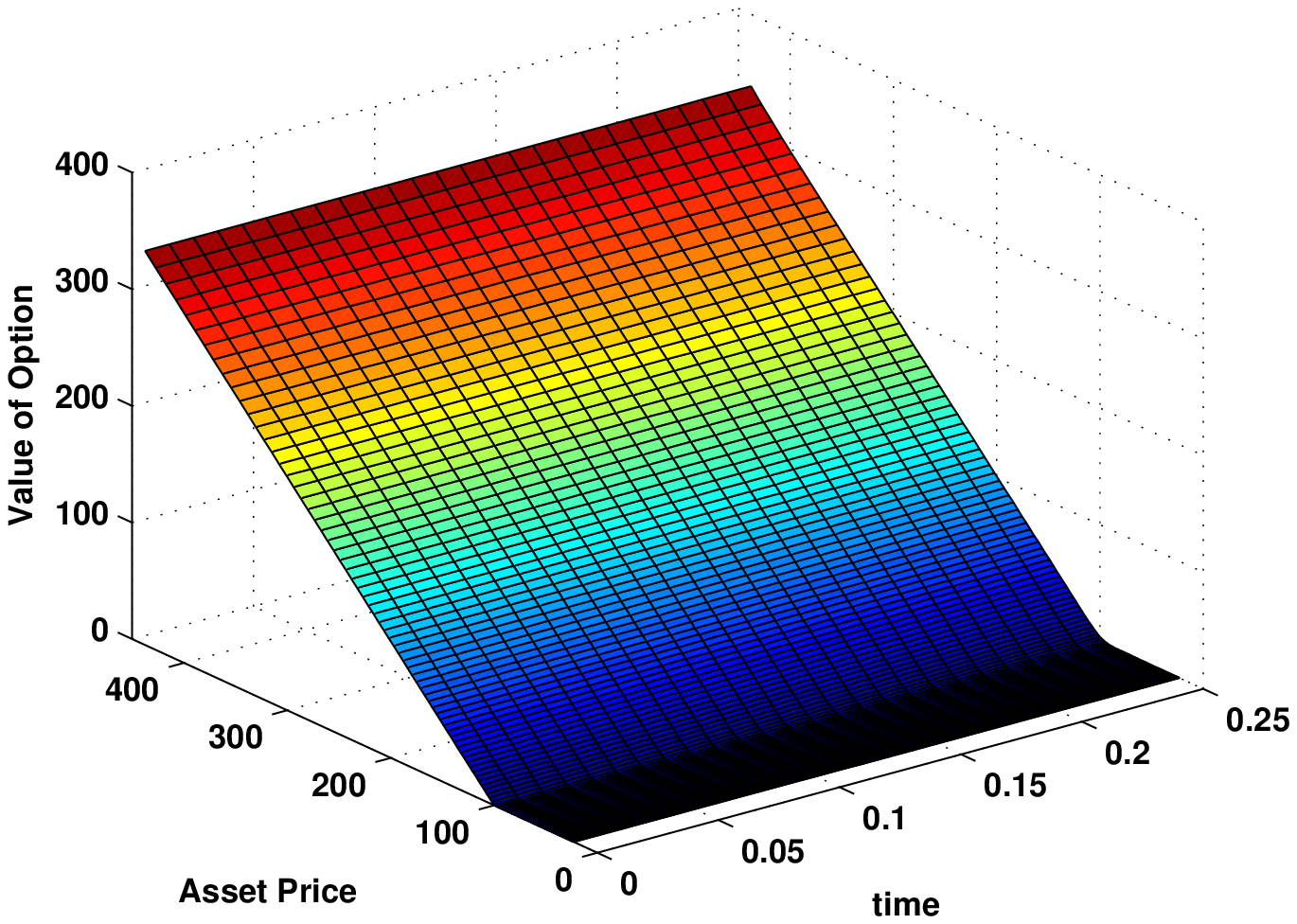}
			\label{fig:kouc_price3d}}
		\caption{Values of European call options using proposed three-time levels compact scheme under Kou jump-diffusion model: (a) As a function of stock price, (b) As a function of stock price and time.}
		\label{fig:kouc}
	\end{center}
\end{figure}
\begin{figure}[h!]
	\centering
	\includegraphics[width=8 cm]{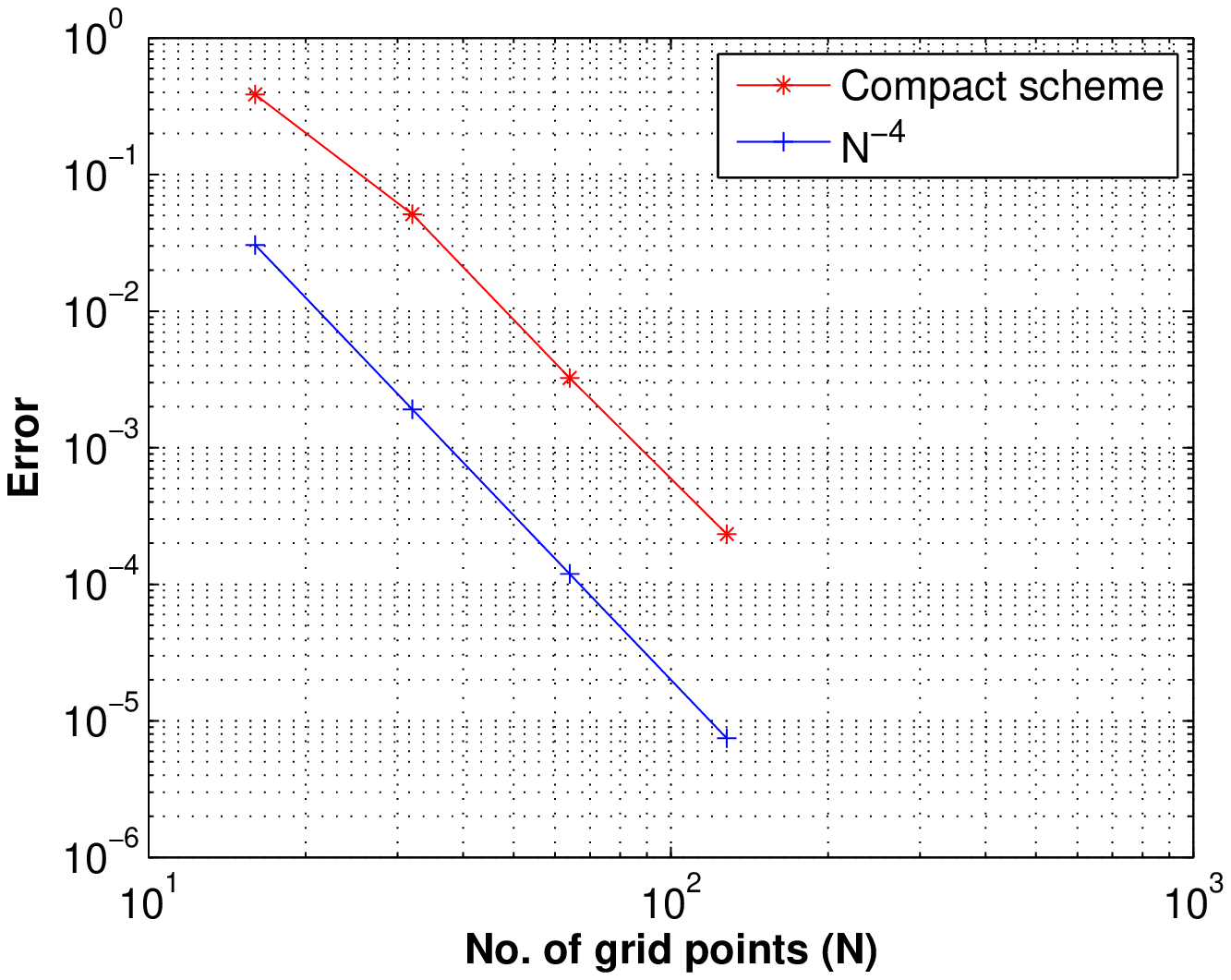}
	\caption{Rate of convergence: Error between the numerical solutions versus number of grid points for European call option under Kou jump-diffusion model.}
	\label{fig:roc_kouc}
\end{figure}
\begin{example}(Merton jump-diffusion model for European put options)
\end{example}
\par The reference values in Table~\ref{table:merton} for European put options under Merton jump-diffusion model are obtained from the infinite series given in \cite{Mer76}. It can be observed from Table~\ref{table:merton} that option prices obtained from the proposed compact scheme for different stock prices are in excellent agreement with the reference values. The initial condition and the numerical solution of the PIDE for European put option under Merton jump-diffusion model are presented in Figure~\ref{fig:price_merton}. Moreover, the values of option as a function of stock price and time are plotted in Figure~\ref{fig:merton_price3d}. The $\ell^2$ error $(Error=\|U(\delta x,\delta \tau)-U(\delta x/2,\delta \tau/2)\|_{\ell^2})$ versus number of grid points is presented in Figure~\ref{fig:roc_mertonp} and it is observed that proposed compact scheme exhibit fourth-order convergence rate. 
\par The analytic solution for European put options under Merton jump-diffusion model is obtained from \cite{Mer76}. The error between the analytic and numerical solution with non-smooth initial condition is presented as a function of time and stock prices in Figures~\ref{fig:err_woutsmo_finite} and~\ref{fig:err_woutsmo_compact}. It is observed from the Figures~\ref{fig:err_woutsmo_finite} and~\ref{fig:err_woutsmo_compact} that proposed compact scheme provides lesser error as compared to finite difference scheme. Similarly, the error between analytic and numerical solution after applying the smoothing operator to the initial condition is plotted as a function of time and stock prices in Figures~\ref{fig:err_smo_finite} and~\ref{fig:err_smo_compact}. It is shown in Figures~\ref{fig:err_smo_finite} and~\ref{fig:err_smo_compact}that less oscillations are produced in the solution near the strike price with the proposed compact scheme.\\
\textbf{Efficiency of proposed compact finite difference scheme}:\\
In order to compare the efficiency of the proposed compact scheme with finite difference scheme, the PIDE~(\ref{eq:pidediscre}) is also solved using finite difference scheme. The error between numerical and analytic solutions in $\ell^2$ norm versus CPU time is presented in Figure~\ref{fig:compu}. It is observed that proposed compact scheme is significantly efficient as compared to the finite difference scheme for a given accuracy.
\begin{example}(Merton jump-diffusion model for European call options)
\end{example}
\par The option prices obtained from the proposed compact scheme and the reference values from \cite{Halluin05} are given in Table~\ref{table:mertonc}. It can be observed from Table~\ref{table:mertonc} that proposed compact scheme is accurate for pricing European call options. The error $(Error=\|U(\delta x,\delta \tau)-U(\delta x/2,\delta \tau/2)\|_{\ell^2})$ versus number of grid points is plotted in Figure~\ref{fig:roc_mertonc} and fourth-order convergence of the proposed compact scheme is shown. The initial condition and numerical solution is presented in Figure~\ref{fig:price_mertonc}. Moreover, the option prices as a function of stock price and time is shown in Figure~\ref{fig:mertonc_price3d}.
\begin{example}(Kou jump-diffusion model for European put options)
\end{example}
\par The option prices obtained from the proposed compact scheme and reference values from \cite{Kwon11} for different stock prices are presented in Table~\ref{table:kou}. It can be observed from the Table~\ref{table:kou} that option prices obtained from the proposed compact scheme are similar to the reference values. The initial condition and numerical solution for European put options is presented in Figure~\ref{fig:price_kou}. The value of option as a function of stock price and time is plotted in Figure~\ref{fig:kou_price3d}. Moreover, fourth-order convergence rate of the proposed compact scheme is observed from Figure~\ref{fig:roc_koup}.
\begin{example}(Kou jump-diffusion model for European call options)
\end{example}
\par The option prices obtained from the proposed compact scheme and reference values from \cite{Halluin05} are presented in Table~\ref{table:kouc}. It can be observed from the  Table~\ref{table:kouc} that option prices obtained from the proposed scheme are in excellent agreement with the reference values. The initial condition and numerical solution for European call option is plotted in Figure~\ref{fig:price_kouc}. The values of option as a function of stock price and time are presented in Figure~\ref{fig:kouc_price3d}. Moreover, it can be observed from Figure~\ref{fig:roc_kouc} that proposed compact scheme exhibit fourth-order convergence rate.
\section*{Acknowledgement} The authors gratefully acknowledge the comments/suggestions of the referee which have greatly improved the paper.
The authors also acknowledges the support provided by Department of Science and Technology, India, under the grant number $SB/FTP/MS-021/2014$.
\addcontentsline{toc}{section}{References}
\bibliographystyle{unsrt}
\bibliography{references}

\begin{thebibliography}{10}

\bibitem{BlaS73}
F.~Black and M.~Scholes.
\newblock The pricing of options and corporate liabilities.
\newblock {\em J. Political Econ.}, 81:637--654, 1973.

\bibitem{Mer76}
R.~C. Merton.
\newblock Option pricing when underlying stock returns are discontinous.
\newblock {\em J. Financial Econ.}, 3:125--144, 1976.

\bibitem{Hull87}
J.~Hull and A.~White.
\newblock The pricing of options on assets with stochastic volatilities.
\newblock {\em J. Finance}, 42:281--300, 1987.

\bibitem{Hes93}
S.~L. Heston.
\newblock A closed form solution for options with stochastic volatility with
  appliacations to bond and currency options.
\newblock {\em Rev. Financial Stud.}, 6:327--343, 1993.

\bibitem{Dup94}
B.~Dupire.
\newblock Pricing with a smile.
\newblock {\em Risk}, 7:18--20, 1994.

\bibitem{Bates96}
D.~Bates.
\newblock Jump and stochastic volatility: exchange rate process implicit in
  deutsche mark options.
\newblock {\em Rev. Finan. Stud.}, 9:69--107, 1996.

\bibitem{Ander00}
L.~Andersen and J.~Andreasen.
\newblock Jump-diffusion process: {V}olatility smile fitting and numerical
  methods for option pricing.
\newblock {\em Rev. Derivatives Res.}, 4:231--262, 2000.

\bibitem{Kou02}
S.~G. Kou.
\newblock A jump-diffusion model for option pricing.
\newblock {\em Manag. Sci.}, 48:1086--1101, 2002.

\bibitem{Briani04}
M.~Briani, C.~L. Chioma, and R.~Natalini.
\newblock Convergence of numerical schemes for viscosity solutions to
  integro-differential degenerate parabolic problems arising in finance theory.
\newblock {\em Numer. Math.}, 98:607--646, 2004.

\bibitem{Cont05}
R.~Cont and E.~Voltchkova.
\newblock A finite difference scheme for option pricing in jump-diffusion and
  exponential {L}evy models.
\newblock {\em SIAM J. Numer. Anal.}, 43:1596--1626, 2005.

\bibitem{Halluin05}
Y.~d'Halluin, P.~A. Forsyth, and K.~R. Veztal.
\newblock Robust numerical methods for contingent claims under jump-diffusion
  process.
\newblock {\em IMA J. Numer. Anal.}, 25:87--112, 2005.

\bibitem{Duffy05}
D.~J. Duffy.
\newblock Numerical analysis of jump--diffusion models: {A} partial
  differential equation approach.
\newblock {\em Technical Report, Datasim}, 2005.

\bibitem{BrNaRu07}
M.~Briani, R.~Natalini, and G.~Russo.
\newblock Implicit–-explicit numerical schemes for jump diffusion processes.
\newblock {\em Calcolo}, 44:33--57, 2007.

\bibitem{SaSt08}
E.~W. Sachs and A.~K. Strauss.
\newblock Efficient solution of a partial integro-differential equation in
  finance.
\newblock {\em Appl. Numer. Math.}, 58:1687--1703, 2008.

\bibitem{Kwon11}
Y.~Kwon and Y.~Lee.
\newblock A second-order finite difference method for option pricing under
  jumps-diffusion models.
\newblock {\em SIAM J. Numer. Anal.}, 49:2598--2617, 2011.

\bibitem{SaTo14}
S.~Salmi and J.~Toivanen.
\newblock {IMEX}-schemes for pricing options under jump--diffusion models.
\newblock {\em Appl. Numer. Math.}, 84:33--45, 2014.

\bibitem{SaToSy14}
S.~Salmi, J.~Toivanen, and L.~V. Sydow.
\newblock An {IMEX}-scheme for pricing options under stochastic volatility
  models with jumps.
\newblock {\em SIAM J. Sci. Comput.}, 36:B817--B834, 2014.

\bibitem{KadT15}
M.~K. Kadalbajoo, L.~P. Tripathi, and Alpesh Kumar.
\newblock Second order accurate {IMEX} methods for option pricing under
  {M}erton and {K}ou jump diffusion model.
\newblock {\em J. Sci.Comput.}, 65:979--1024, 2015.

\bibitem{Lele92}
S.~K. Lele.
\newblock Compact finite difference schemes with spectral-like resolution.
\newblock {\em J. Comput. Phys.}, 103:16--42, 1992.

\bibitem{DurF04}
B.~During, M.~Fournie, and A.~Jungel.
\newblock Convergence of high-order compact finite difference scheme for a
  nonlinear {B}lack-{S}choles equation.
\newblock {\em Math. Model. Numer. Anal.}, 38:359--369, 2004.

\bibitem{TanGB08}
D.~Y. Tangman, A.~Gopaul, and M.~Bhuruth.
\newblock Numerical pricing of options using high-order compact finite
  difference schemes.
\newblock {\em J. Comput. Appl. Math}, 218:270--280, 2008.

\bibitem{HWSun11}
S.~T. Lee and H.~W. Sun.
\newblock Fourth order compact scheme with local mesh refinement for option
  pricing in jump-diffusion model.
\newblock {\em Numer Methods Partial Differ Equ.}, 28:1079--1098, 2011.

\bibitem{DuPi17}
B.~During and A.~Pitkin.
\newblock High-order compact finite difference scheme for option pricing in
  stochastic volatility jump models.
\newblock {\em arXiv:1704.05308v1}.

\bibitem{FasM00}
H.~L. Meitz and H.~F. Fasel.
\newblock A compact-difference scheme for the {N}avier–-{S}tokes equations in
  vorticity-velocity formulation.
\newblock {\em J. Comput. Phys.}, 157:371, 2000.

\bibitem{TKSen03}
T.~K. Sengupta, G.~Ganeriwal, and S.~De.
\newblock Analysis of central and upwind compact schemes.
\newblock {\em J. Comput. Phys.}, 192:677, 2003.

\bibitem{MKPACM17}
M.~Mehra and K.~S. Patel.
\newblock Algorithm 986: A suite of compact finite difference schemes.
\newblock {\em ACM Trans. Math. Softw.}, 44, 2017.

\bibitem{thomee70}
H.~O. Kreiss, V.~Thomee, and O.~Widlund.
\newblock Smoothing of initial data and rates of convergence for parbolic
  difference equations.
\newblock {\em Comm. Pure Appl. Math.}, 23:241--259, 1970.

\bibitem{KPMI17}
K.~S. Patel and M.~Mehra.
\newblock Fourth-order compact finite difference scheme for {A}merican option
  pricing under regime-switching jump-diffusion models.
\newblock {\em Int. J. Appl. Comput. Math.}, 3:547--567, 2017.

\bibitem{KPMA17}
K.~S. Patel and M.~Mehra.
\newblock A numerical study of {A}sian option with high-order compact finite
  difference scheme.
\newblock {\em J. Appl. Math. Comput.}, 2017.
\newblock \mbox{DOI}:10.1007/s12190-017-1115-2.

\bibitem{Chan07}
R.~Chan and X.~Jin.
\newblock {\em An introduction to iterative Toepliz solvers}.
\newblock SIAM, 2007.

\bibitem{Chan96}
R.~Chan and M.~Ng.
\newblock Conjugate gradient methods for toeplitz systems.
\newblock {\em SIAM Rev.}, 38:427--482, 1996.

\bibitem{Lambert91}
J.~D. Lambert.
\newblock {\em Numerical Methods for Ordinary Differential Systems: The Initial
  Value Problem}.
\newblock John Wiley and Sons, 1991.

\bibitem{KPMD17}
K.~S. Patel and M.~Mehra.
\newblock High-order compact finite difference scheme for pricing {A}sian
  option with moving boundary condition.
\newblock {\em Differ Equ Dyn Syst}, 2017.
\newblock \mbox{DOI}:10.1007/s12591-017-0372-8.

\bibitem{Jcstrik04}
J.~C. Strikewerda.
\newblock {\em Finite Difference Schemes and Partial Differential Equations}.
\newblock SIAM, 2004.

\end{thebibliography}
\end{document}